\documentclass{article}
\usepackage{authblk}

\usepackage{amsmath}
\usepackage{amsthm}
\usepackage{amssymb}
\usepackage{latexsym}
\usepackage{colonequals}
\usepackage[inline]{enumitem}
\usepackage{quiver}
\usepackage{hyperref}

\usepackage[capitalize]{cleveref}
\usepackage{thmtools}

\usepackage{virginialake}
\vlnosmallleftlabels
\vlnostructuressyntax

\usepackage{ifthen}
\usepackage{xparse}

\newcommand{\paaras}[1]{}
\newcommand{\silvia}[1]{}

\newcommand{\AND}{\mathbin{\wedge}}
\newcommand{\OR}{\mathbin{\vee}}
\newcommand{\IMPLIES}{\mathbin{\supset}}
\newcommand{\IMP}{\IMPLIES}

\newcommand{\BOX}[1][]{\Box_{#1}}
\newcommand{\bottom}{\mathord{\bot}}
\newcommand{\BOT}{\mathord{\bottom}}

\newcommand{\BIGAND}{\bigwedge}

\newcommand{\BNFequals}{\coloncolonequals}

\newcommand{\langjust}{\mathcal{L}_{\J}}
\newcommand{\langpl}{\mathcal{L}_{\IPL}}
\newcommand{\langipl}{\langpl}
\newcommand{\J}{\mathsf{J}_{\lambda}}
\newcommand{\IPL}{\mathsf{IPL}}
\newcommand{\NDIPL}{\lambda\mathsf{IPL}}
\newcommand{\LP}{\mathsf{LP}}
\newcommand{\iLP}{\mathsf{iLP}}
\newcommand{\sfour}{\mathsf{S4}}
\newcommand{\PA}{\mathsf{PA}}
\newcommand{\HA}{\mathsf{HA}}
\newcommand{\axrule}[1]{\mathsf{#1}}
\newcommand{\iplax}[1]{\axrule{PL}_{#1}}
\newcommand{\jk}{\axrule{jk}}
\newcommand{\jimpi}{\axrule{j}_{\intrle{\IMP}}}
\newcommand{\jandi}{\axrule{j}_{\intrle{\AND}}}
\newcommand{\jandel}{\axrule{j}_{\elimrleleft{\AND}}}
\newcommand{\jander}{\axrule{j}_{\elimrleright{\AND}}}

\newcommand{\intrle}[1]{#1_\mathsf{I}}

\newcommand{\elimrle}[1]{#1_\mathsf{E}}
\newcommand{\elimrleleft}[1]{#1_\mathsf{E_l}}
\newcommand{\elimrleright}[1]{#1_\mathsf{E_r}}
\newcommand{\jtax}{\axrule{jt}}
\newcommand{\jfax}{\axrule{j4}}
\newcommand{\Prop}{\mathsf{Prop}}

\newcommand{\pv}[1][]{{x_{#1}}}
\newcommand{\pdot}[2]{#1 \mathbin{\cdot} #2}
\newcommand{\plambda}[2][\pv]{\lambda #1 . #2}
\newcommand{\pl}[2][x]{\plambda[#1]{#2}}
\newcommand{\pbang}[1]{\mathord{!}#1}
\newcommand{\pael}[1]{\pi_l(#1)}
\newcommand{\paer}[1]{\pi_r(#1)}
\newcommand{\pai}[2]{\langle #1 , #2 \rangle}

\newcommand{\prfvar}{\mathsf{PVar}}
\newcommand{\prfterms}{\mathsf{PTerms}}
\newcommand{\terms}{\lambda\mathsf{PTerms}}
\newcommand{\iplterms}{\mathsf{IPLTerms}}
\newcommand{\IPLterms}{\iplterms}
\newcommand{\Vars}{\lambda \mathsf{Var}}

\newcommand{\boxterm}[1]{\BOX (#1)}
\newcommand{\just}[2]{[#1]{#2}}

\newcommand{\hprove}[2][]{#1 \vdash_{\J} #2}
\newcommand{\hiplprove}[2][]{#1 \vdash_{\IPL} #2}
\renewcommand{\mp}{\mathsf{mp}}

\newcommand{\sset}{\subseteq}
\newcommand{\set}[2][]{%
	\ifthenelse{\equal{#1}{}}%
	{\{#2\}}%
	{\{#1 \mid #2\}}%
}
\newcommand{\function}[3]{#1 : #2 \rightarrow #3}

\newcommand{\UNI}{\cup}

\newcommand{\ptttsymb}{*}
\newcommand{\pttt}[1]{#1^{\ptttsymb}}
\newcommand{\ttptsymb}{\star}
\newcommand{\ttpt}[1]{#1^{\ttptsymb}}

\newcommand{\ctxfm}[2]{#1 {:} #2}
\newcommand{\fm}[2]{\ctxfm{#1}{#2}}
\newcommand{\iplctxvdash}{\vdash_{\NDIPL}}
\newcommand{\iplctx}[2]{#1 \iplctxvdash #2}
\newcommand{\iplctxset}[1][\Gamma]{#1}
\newcommand{\NDJ}{\mathsf{N} \J}
\newcommand{\utctxvdash}{\vdash_{\NDJ}}
\newcommand{\utctx}[2]{#1 \utctxvdash #2}
\newcommand{\nutctx}[2]{#1 \utctxvdash^{\mathsf{n}} #2}
\newcommand{\laJ}{\lambda \J}
\newcommand{\ctxvdash}{\vdash_{\laJ}}
\newcommand{\ctx}[2]{#1 \ctxvdash #2}
\newcommand{\nctx}[2]{#1 \ctxvdash^{\mathsf{n}} #2}

\newcommand{\var}[1][]{a_{#1}}
\newcommand{\app}[2]{#1 #2}
\newcommand{\la}[2][\var]{\lambda #1 . #2}
\newcommand{\projl}[1]{\pi_l (#1)}
\newcommand{\projr}[1]{\pi_r (#1)}
\newcommand{\andc}[2]{\langle #1, #2 \rangle}
\newcommand{\bang}[1]{{!}#1}
\newcommand{\prom}[2]{P_{#1}(#2)}
\newcommand{\boxapp}[2]{A^{\BOX}(#1, #2)}
\newcommand{\boxla}[2][\var]{\lambda^{\BOX} #1 . #2}
\newcommand{\boxprojl}[1]{\pi_l^{\BOX} (#1)}
\newcommand{\boxprojr}[1]{\pi_r^{\BOX} (#1)}
\newcommand{\boxandc}[2]{\langle #1, #2 \rangle^{\BOX}}
\newcommand{\unbox}[1]{U(#1)}
\newcommand{\typset}[1]{\overrightarrow{#1}}
\newcommand{\subst}[3]{#1[#2 \colonequals #3]}

\vlnosmallleftlabels

\newcommand{\vlstr}[3]{\vltr{#1}{#2}{\vlhy{}}{#3}{\vlhy{}}}
\newcommand{\vlhtr}[2]{\vlstr{#1}{#2}{\vlhy{\hskip1.5em}}}
\newcommand{\seqarrow}{\Rightarrow}
\newcommand{\seq}[2]{#1 \seqarrow #2}
\newcommand{\form}[1]{\text{form}(#1)}
\newcommand{\seqprove}[1]{\vdash_{\LJ} #1}
\newcommand{\seqprovecut}[1]{\vdash_{\LJ + \cut} #1}
\newcommand{\rr}[1]{#1_{\mathsf{r}}}
\newcommand{\lr}[1]{#1_{\mathsf{l}}}
\newcommand{\id}{\mathsf{id}}
\newcommand{\gid}{\mathsf{gid}}
\newcommand{\cont}{\mathsf{c}}
\newcommand{\boxrule}[1][]{[#1]}
\newcommand{\cut}{\mathsf{cut}}
\newcommand{\LJ}{\mathsf{L}\J}
\newcommand{\rle}{\mathsf{r}}
\newcommand{\ann}[1]{#1^{\annotationsymb}}
\newcommand{\annotationsymb}{\downarrow}
\newcommand{\rank}[2]{\text{rank}(#1, #2)}
\renewcommand{\deg}[1]{\text{deg}(#1)}
\renewcommand{\max}[2]{\text{max}(#1, #2)}
\newcommand{\cutrank}[1]{\text{cutrank}(#1)}
\newcommand{\height}[1]{\text{h}(#1)}

\newtheorem{theorem}{Theorem}
\newtheorem{lemma}[theorem]{Lemma}
\newtheorem{proposition}[theorem]{Proposition}

\theoremstyle{definition}
\newtheorem{definition}[theorem]{Definition}
\newtheorem{example}[theorem]{Example}

\theoremstyle{remark}
\newtheorem{remark}[theorem]{Remark}

\makeatletter
\def\@maketitle{%
	\newpage
	\null
	\vskip 2em%
	\begin{center}%
		\let \footnote \thanks
		{\Large\bfseries \@title \par}%
		\vskip 1.5em%
		{\normalsize
			\lineskip .5em%
			\begin{tabular}[t]{c}%
				\@author
			\end{tabular}\par}%
		\vskip 1em%
		{\normalsize \@date}%
	\end{center}%
	\par
	\vskip 1.5em}
\makeatother

\title{Justification Logic of the Lambda Calculus\\ 
}

\author{Silvia Ghilezan}
\affil{Faculty of Technical Sciences, University of Novi Sad
\\
Mathematical Institute of the Serbian Academy of Sciences and Arts}

\author{Paaras Padhiar}
\affil{School of Computer Science, University of Birmingham}

\date{}

\begin{document}

\maketitle

\begin{abstract}
The simply typed $\lambda$-calculus is a model of computation where typed terms correspond to proofs of intuitionistic propositional logic ($\IPL$) via the Curry–Howard correspondence. 
Justification logic is an operational modal logic in which the standard box modality is replaced by an explicit proof term, allowing the logic itself to reason directly about proofs of its formulas. 
Standard justification logics reason about proofs of $\IPL$ by embedding Hilbert-style axiomatic proofs of $\IPL$ as proof terms of the logic. 
We instead introduce a justification logic of the $\lambda$-calculus, in which the proof terms of the modality are exactly the typed $\lambda$-terms themselves: a modal logic that reasons about computation and proof simultaneously, as both notions coincide under the Curry–Howard interpretation.

First, we provide an axiomatisation of this logic. 
We then propose a natural deduction system and a Curry–Howard interpretation, where a formal connection between the term calculus and the proof terms of the logic is provided. 
To complete the picture, we give a Gentzen-style sequent calculus for which we prove cut-elimination, and consequently obtain a normalisation result by working in the negative fragment.
\end{abstract}

\section{Introduction}
Justification logic is an explicit modal logic where the $\BOX$ modality is replaced with an explicit proof term $t$.
The modal formula $\BOX A$ read as `$A$ is \emph{provable}' is replaced with $\just{t}{A}$ interpreted as `$t$ is a \emph{proof} of $A$'.
Proof terms are then manipulated through \emph{proof operations}.
In standard justification logics, the usual operations include the \emph{proof application} operation which \emph{internalises} modus ponens.
The first justification logic is the \emph{logic of proofs} ($\LP$) introduced in~\cite{artemov_operational_1995,artemov_explicit_2001} which is an explicit counterpart to the classical modal logic $\sfour$.
Via the G{\"o}del translation of intuitionistic propositional logic ($\IPL$) into $\sfour$, $\LP$ can also be read as a provability semantics for $\IPL$. 
Under this reading, Hilbert-style proofs of $\IPL$ embed directly into $\LP$: 
a theorem $A$ of $\IPL$ is translated to a theorem $\just{t}{A}$ of $\LP$ for some proof term $t$. 
These proof terms can in turn be constructed through the lens of combinatory logic, where they correspond directly to manipulations of axiomatic reasoning and modus ponens~\cite{shamkanov_strong_2011}.

With a similar motivation, we propose a \emph{Justification Logic of the $\lambda$-calculus}, $\J$, an intuitionistic logic, where we incorporate terms of the $\lambda$-calculus of the negative fragment, more specifically the $(\AND, \IMP)$-fragment, of $\IPL$\footnote{This is also a fragment of minimal propositional logic but we will not be making such a distinction.}.
We do so through embedding proofs of the $\lambda$-calculus into the logic, similarly to how Hilbert proofs of $\IPL$ are embedded into $\LP$.

A feature of justification logic is \emph{internalisation}, where the logic can internally reason about its own proofs.
We utilise this for $\J$ where the logic is then able to reason about its own computation.
As a natural next step, we then explore a natural deduction formulation of $\J$ to explore the computational content of that system and see how such content and proof terms relate.

The quest to explore a connection between proof terms and $\lambda$-terms began with
Artemov~\cite{artemov_unified_2002} where $\LP$ was studied on an intuitionistic base giving a Gentzen-style sequent calculus and natural deduction system.
The computational content of the resulting intuitionistic logic of proofs ($\iLP$) was investigated through a Curry-Howard interpretation, although only a limited argument for normalisation was given.
In~\cite{artemov_intensional_2007} strong normalisation for ($\iLP$) is achieved by using an extended natural deduction system which traces `historical' computation through judgements.
Other attempts include~\cite{alt_reflective_2001}, using ideas of the $\lambda$-calculus, and~\cite{deboer_justification_2022,deboer_justification_2022-1} through the lens of type theory.

The system in~\cite{artemov_intensional_2007} is later utilised in~\cite{bavera_justification_2010,steren_intuitionistic_2014,bavera_justification_2018,ricciotti_strongly_2017} to study of other justification logics of computations.
Of note, 
\cite{steren_intuitionistic_2014} studies the \emph{Hypothetical Logic of Proofs}, a logic of natural deduction judgements of justification logic 
in which terms of the $\lambda$-calculus of the traced computation are incorporated into the proof terms of the justification logic. 
The obstacle that recurs throughout this line of work is that the $\lambda$-calculus is strictly more expressive than proof terms of $\LP$: establishing a connection requires notions of term equality that have no direct counterpart on the justification logic side.
As our logic is distinct to $\LP$, and is constructed directly from the $\lambda$-calculus, we do not need to make such considerations.

\textbf{Contributions and outline of the paper}.
\cref{sec:JL,sec:ND,sec:pt} are our contributions to the literature.
\cref{sec:ipl} recalls the basic facts of the $(\AND, \IMP)$-fragment of $\IPL$ along with its $\lambda$-calculus.
\cref{sec:JL} introduces $\J$ axiomatically, discuss its internalisation features and how the $\lambda$-calculus \emph{embeds} into the logic.
\cref{sec:ND} introduces the untyped $\NDJ$ and typed $\laJ$ natural deduction calculi.
In that section, we also relate the proof terms of the language of $\J$ to the term calculus of $\laJ$,
\cref{sec:pt} introduces the sequent calculus formulation of $\J$, $\LJ$, where $\LJ + \cut$ is the system with the \emph{cut} rule, where a cut-elimination argument is made, and is used to achieve \emph{normalisation}.
These results are visualised as a `tour' in \cref{fig:results}.

\begin{figure}[t]
    \centering
\[\begin{tikzcd}
	&&&& \laJ && {\text{Normal $\laJ$}} \\
	\\
	&&&& \NDJ && {\text{Normal $\NDJ$}} \\
	\\
	{\text{$\lambda$-calculus}} && \J && {\LJ + \cut} && \LJ
	\arrow[dashed, from=1-5, to=1-7]
	\arrow["{\text{\cref{prop:jlcurryhoward}}}", tail reversed, from=3-5, to=1-5]
	\arrow[dashed, from=3-5, to=3-7]
	\arrow["{\text{\cref{lem:NDtoSC}}}", from=3-5, to=5-5]
	\arrow["{\text{\cref{prop:jlcurryhoward}}}"', tail reversed, from=3-7, to=1-7]
	\arrow["{\text{\cref{prop:iplintoj}}}"', hook, from=5-1, to=5-3]
	\arrow["{\text{\cref{thm:ndsoundcomp}}}", tail reversed, from=5-3, to=3-5]
	\arrow["{\text{\cref{prop:seqcomplete}}}"', from=5-3, to=5-5]
	\arrow["{\text{\cref{thm:cutelim}}}"', from=5-5, to=5-7]
	\arrow["{\text{\cref{lem:SCtoND}}}"', from=5-7, to=3-7]
	\arrow["{\text{\cref{prop:soundness}}}", shift left=4, curve={height=-30pt}, from=5-7, to=5-3]
\end{tikzcd}\]
    \caption{Tour of results}
    \label{fig:results}
\end{figure}

\section{Intuitionistic Propositional Logic}\label{sec:ipl}
To begin, we recall the axiomatisation of the $(\AND, \IMP)$-fragment of intuitionistic propositional logic~($\IPL$) and its typed $\lambda$-calculus.

\subsection{Axiomatisation}
Formulas of the language~$\langpl$ of~$\IPL$ are given by the grammar:
$$
    A
    \BNFequals
    p
    \mid
    (A \AND A)
    \mid
    (A \IMP A)
$$
where $p$ ranges over a countable set of propositional atoms $\Prop$.
We assume~$\IMP$ is right associative and drop the outermost brackets when it is clear to do so.

$\IPL$ is given by the following axiom schema
$$
    \begin{array}{rclrcl}
        \iplax{1} & {:} & A \IMP (B \IMP A) &
        \iplax{2} & {:} & (A \IMP (B \IMP C)) \IMP (A \IMP B) \IMP (A \IMP C) \\
        \iplax{3} & {:} & A \IMP B \IMP (A \AND B) &
        \iplax{4} & {:} & (A \AND B) \IMP A \\
        \iplax{5} & {:} & (A \AND B) \IMP B & & &
    \end{array}
$$
where consecutions $\hiplprove[\Gamma]{A}$ for $\Gamma \sset \langpl$ and $A \in \langpl$ is defined inductively as follows:
$$
    \vlinf{}{}
    {
        \hiplprove[\Gamma]{A}
    }
    {
        \text{$A$ is an axiom instance}
    }
    \quad
    \vlinf{}{}
    {
        \hiplprove[\Gamma]{A}
    }
    {
        A \in \Gamma
    }
    \quad
    \vliinf{\mp}{}
    {
        \hiplprove[\Gamma]{B}
    }
    {
        \hiplprove[\Gamma]{A}
    }
    {
        \hiplprove[\Gamma]{A \IMP B}
    }
$$
We say $A$ is a \emph{theorem} of $\IPL$ if $\hiplprove{A}$ as usual.

\subsection{Typed $\lambda$-Calculus}
The proof terms of the annotated natural deduction are defined by the grammar:
$$
    t 
    \BNFequals 
    \pv 
    \mid 
    \pdot{t}{t}
    \mid
    \pl{t}
    \mid
    \pael{t}
    \mid
    \paer{t}
    \mid
    \pai{t}{t}
$$
where $x$ ranges over a countable set of proof variables $\prfvar$.
A variable~$\pv$ in~$t$ is \emph{free} if it is not contained in any subterm of $t$ of the form $\pl{s}$.
We denote the set of these terms as~$\IPLterms$.
A term is \emph{closed} if it does not contain any free variables.
We write~$t(\pv[1], \dots, \pv[n])$ if~$t$ is a proof term constructed from variables~$\pv[1], \dots, \pv[n]$.
We \emph{type} formulas~$A \in \langpl$ as standard with proof terms $t$ and write $\ctxfm{t}{A}$.

\begin{definition}
    An $\IPL$-\emph{judgement} is written $\iplctx{\typset{\Gamma}}{\ctxfm{t}{A}}$, where $\typset{\Gamma}$ is a set of typed formulas~$\ctxfm{\pv[i]}{A_i}$ where~$\pv[i] \in \prfvar$ is defined inductively by
    the natural deduction calculus~$\NDIPL$ given by the rules in \cref{fig:iplnd} with~$\id$ as the base case, and a \emph{proof} in $\NDIPL$ is constructed as trees from these rules with root $\iplctx{\typset{\Gamma}}{\ctxfm{t}{A}}$ and leaves closed by $\id$.
\end{definition}
\begin{figure}[t]
    \fbox
    {
        \begin{minipage}{.95\textwidth}
            \centering
            $
                \vlinf{\id}{}
                {
                    \iplctx{\iplctxset, \ctxfm{\pv}{A}}{\ctxfm{\pv}{A}}
                }
                {}
            $
            \\[1ex]
            $
                \vlinf{\intrle{\IMP}}{}
                {
                    \iplctx{\iplctxset}{\ctxfm{\pl{t}}{A \IMP B}}
                }
                {
                    \iplctx{\iplctxset, \ctxfm{\pv}{A}}{\ctxfm{t}{B}}
                }
                \qquad
                \vliinf{\elimrle{\IMP}}{}
                {
                    \iplctx{\iplctxset}{\ctxfm{\pdot{s}{t}}{B}}
                }
                {
                    \iplctx{\iplctxset}{\ctxfm{s}{A \IMP B}}
                }
                {
                    \iplctx{\iplctxset}{\ctxfm{t}{A}}
                }
            $
            \\[1ex]
            $
                \vliinf{\intrle{\AND}}{}
                {
                    \iplctx{\iplctxset}{\ctxfm{\pai{s}{t}}{A \AND B}}
                }
                {
                    \iplctx{\iplctxset}{\ctxfm{s}{A}}
                }
                {
                    \iplctx{\iplctxset}{\ctxfm{t}{B}}
                }
                \qquad
                \vlinf{\elimrleleft{\AND}}{}
                {
                    \iplctx{\iplctxset}{\ctxfm{\pael{t}}{A}}
                }
                {
                    \iplctx{\iplctxset}{\ctxfm{t}{A \AND B}}
                }
                \qquad
                \vlinf{\elimrleright{\AND}}{}
                {
                    \iplctx{\iplctxset}{\ctxfm{\paer{t}}{B}}
                }
                {
                    \iplctx{\iplctxset}{\ctxfm{t}{A \AND B}}
                }
            $
        \end{minipage}
    }
    \caption{Natural deduction calculus $\NDIPL$}
    \label{fig:iplnd}
\end{figure}
By the Curry-Howard correspondence, we have the following (see \cite{sorensen_lectures_2006}).
\begin{theorem}
    Let $\Gamma \sset \langpl$ and $A \in \langpl$.
    We write $\typset{\Gamma}$ typing the formulas $\Gamma$ with distinct proof variables.
    Then
    \begin{enumerate}
        \item If $\hiplprove[\Gamma]{A}$, then $\iplctx{\typset{\Gamma}}{\ctxfm{t}{A}}$ for some proof term $t$ constructed from proof variables contained in $\typset{\Gamma}$.
        \item If $\iplctx{\typset{\Gamma}}{\ctxfm{t}{A}}$, then $\hiplprove[\Gamma]{A}$.
    \end{enumerate}
\end{theorem}

\section{Justification Logic $\J$}\label{sec:JL}
In this section, we introduce our justification logic~$\J$ for $\NDIPL$ axiomatically with a discussion of its internalization features.
We also briefly discuss the relations to other justification logics in the literature which include forms of $\lambda$-abstraction in the language of the logic.

\subsection{Axiomatisation}
We extend \emph{proof terms} of $\NDIPL$ with the $\pbang{}$ operator:
$$
    t 
    \BNFequals 
    \pv 
    \mid 
    \pdot{t}{t}
    \mid
    \pl{t}
    \mid
    \pael{t}
    \mid
    \paer{t}
    \mid
    \pai{t}{t}
    \mid
    \pbang{t}
$$
where $x$ ranges over a countable set of proof variables $\prfvar$.
We denote the set of proof terms as~$\prfterms$ and note~$\iplterms \sset \prfterms$.
Free variables and closed terms are similarly defined.
\emph{Formulas} of the language~$\langjust$ are given by the grammar
$$
    A
    \BNFequals
    p
    \mid
    (A \AND A)
    \mid
    (A \IMP A)
    \mid
    \just{t}{A}
$$
where $p$ ranges over a countable set of propositional atoms $\Prop$.

The logic $\J$ is defined by extending intuitionistic propositional logic ($\IPL$) with:
$$
    \begin{array}{rclrcl}
         \jk & {:} & \just{s}{(A \IMP B)} \IMP (\just{t}{A} \IMP \just{\pdot{s}{t}}{B})  &
         \jimpi & {:} & (\just{x}{A} \IMP \just{t}{B}) \IMP \just{\pl{t}}{(A \IMP B)} \\
         \jandi & {:} & \just{s}{A} \IMP (\just{t}{B} \IMP \just{\pai{s}{t}}{(A \AND B)}) &
         \jandel & {:} & \just{t}{(A \AND B)} \IMP \just{\pael{t}}{A} \\
         \jander & {:} & \just{t}{(A \AND B)} \IMP \just{\paer{t}}{B} &
         \jtax & : & \just{t}{A} \IMP A \\
         \jfax & : & \just{t}{A} \IMP \just{\pbang{t}}{\just{t}{A}}
         & & &
    \end{array}
$$
Consecutions $\hprove[\Gamma]{A}$ where $\Gamma \sset \langjust$ and $A \in \langjust$ is defined inductively as follows:
$$
    \vlinf{}{}
    {
        \hprove[\Gamma]{A}
    }
    {
        \text{$A$ is an axiom instance}
    }
    \quad
    \vlinf{}{}
    {
        \hprove[\Gamma]{A}
    }
    {
        A \in \Gamma
    }
    \quad
    \vliinf{\mp}{}
    {
        \hprove[\Gamma]{B}
    }
    {
        \hprove[\Gamma]{A}
    }
    {
        \hprove[\Gamma]{A \IMP B}
    }
$$
We say $A$ is a \emph{theorem} of $\J$ if $\hprove{A}$.

For shorthand, given $\Delta \sset \langjust$ and $B \in \langjust$, we write
$
    \hprove[\Gamma, B]{A}
$
and
$
    \hprove[\Gamma, \Delta]{A}
$
for
$
    \hprove[\Gamma \UNI \set{B}]{A}
$
and
$
    \hprove[\Gamma \UNI \Delta]{A}
$.

\begin{remark}
    The logic restricted to $\jk$, $\jtax$ and $\jfax$ is a fragment of the intuitionistic Logic of Proofs~\cite{artemov_unified_2002}.
    The axiom $\jk$ \emph{internalises} modus ponens at the level of the modality, $\jtax$ gives an \emph{provability} reading of the modality and $\jfax$ allows for reasoning about proofs at a higher level where the $\pbang{}$ operator denotes \emph{positive introspection} or \emph{proof checking}.
    The other axioms along with~$\jk$ reason about the introduction and elimination rules of $\NDIPL$.
\end{remark}

\subsection{Properties and Internalisation of Justification Logic}
As is standard in justification logic, we have the following
\begin{theorem}[Deduction Theorem]\label{thm:deduction}
    Let $\Gamma \sset \langjust$ and $A \in \langjust$.
    Then
    $$
        \hprove[\Gamma, B]{A} \iff \hprove[\Gamma]{B \IMP A}.
    $$
\end{theorem}
As $\J$ doesn't require some form of necessitation rule that is usually required for modal logics, it avoids the usual considerations surrounding the Deduction Theorem~\cite{hakli_does_2012}.
However like standard justification logics, $\J$ is able to internalise its own reasoning similar to a necessitation rule.
\begin{theorem}[Self-internalisation]\label{thm:internalisation}
    If $\hprove{A}$, then there exists a \emph{closed} proof term $t$ such that $\hprove{\just{t}{A}}$. 
\end{theorem}
\begin{proof}
    We proceed by induction on $\hprove A$.
    For the base cases we consider when $A$ is an axiom instance.
    For the full proof, see \cref{sec:appjl} of the Appendix.
    We highlight the following base cases:
    
    Case: $A$ is an instance of the $\jk ,\jimpi, \jandi , \jandel , \jander ,\jfax$, we have
    $$
        A \colonequals {A_1} \IMP {A_2} \IMP \dots \IMP A_{n-1} \IMP \just{t_n}{A_n} 
    $$
    for some proof term $t_n$ and formulas $A_1, \dots, A_n$.
    For each $i = 1, \dots, n$ we have
    $$
        \hprove[\just{x_1}{A_1}, \just{x_2}{{A_2}}, \dots \just{x_{n-1}}{A_{n-1}}]
        {\just{x_i}{{A_i}}}
    $$
    Using the $\jtax$ axiom, we get
    $
        \hprove[\just{x_1}{A_1}, \just{x_2}{{A_2}}, \dots \just{x_{n-1}}{A_{n-1}}]
        {{A_i}}
    $.
    Using the axiom $A$ and $\mp$ repeatedly, we get
    $$
        \hprove[\just{x_1}{A_1}, \just{x_2}{{A_2}}, \dots \just{x_{n-1}}{A_{n-1}}]
        {
            \just{t_n}{A_n}
        }
    $$
    Using the $\jfax$ axiom, we have
    $$
        \hprove[\just{x_1}{A_1}, \just{x_2}{{A_2}}, \dots \just{x_{n-1}}{A_{n-1}}]
        {
            \just{\pbang{t_n}}{\just{t_n}{A_n}}
        }
    $$
    Applying the Deduction Theorem (\cref{thm:deduction}) and $\jimpi$ several times, we get
    $$
        \hprove
        {
            \just{\pl[x_1]{\pl[x_2]{ \dots \pl[x_{n-1}}}]{\pbang{t_n}}}{({A_1} \IMP {A_2} \IMP \dots \IMP A_{n-1} \IMP \just{t_n}{A_n} )}
        }
    $$
    as desired.
\end{proof}

\begin{remark}\label{rem:iplthms}
    In the proof of \cref{thm:internalisation}, when $A$ is a theorem of $\IPL$, the proof term obtained is precisely a proof term where $\iplctx{}{\fm{t}{B}}$ where $B$ is a formula of the same shape as $A$ in the language of $\langipl$.
\end{remark}

Self-internalisation can be extended to lifting assumptions with proof terms as well.
\begin{theorem}[Lifting Lemma]\label{thm:lifting}
    If $\hprove[B_1, \dots, B_n]{A}$, then for proof variables $\pv[1], \dots, \pv[n]$, there exists a proof term $t(\pv[1], \dots, \pv[n])$ such that
    $$
        \hprove[\just{\pv[1]}{B_1}, \dots, \just{\pv[n]}{B_n}]{\just{t(\pv[1], \dots, \pv[n])}{A}}.
    $$
\end{theorem}
\begin{proof}
    Similar to the proof of Self-internalisation (\cref{thm:internalisation}), this proceeds by induction on $\hprove[B_1, \dots, B_n]{A}$.

    For the base case, if $A$ is an axiom of $\J$, then construct the proof term using \cref{thm:internalisation}.
    Otherwise, $A = B_i$ for some $i = 1, \dots, n$ and set $t(x_1, \dots, x_n) \colonequals x_i$.

    The inductive case of a conclusion of $\mp$, follows the same argument as in the proof of \cref{thm:internalisation}.
\end{proof}

We should note in the Self-internalisation Theorem (\cref{thm:internalisation}) and Lifting Lemma (\cref{thm:lifting}) that when the formulas are restricted to the language of~$\IPL$ ($\langpl$), then the proof term~$t$ obtained is precisely a proof term which is obtained in the natural deduction calculus for~$\IPL$ ($\NDIPL$).
This is formulated precisely in the following.

\begin{proposition}\label{prop:iplintoj}
    Let $A_1, \dots, A_n, A \in \langpl$
    and $\pv[1], \dots, \pv[n] \in \prfvar$.
    Then for some proof term~$t(\pv[1], \dots, \pv[n])$,
    $$
        \iplctx{\ctxfm{\pv[1]}{A_1}, \dots, \ctxfm{\pv[n]}{A_n}}{\ctxfm{t}{A}}
        \implies
        \hprove[\just{\pv[1]}{A_1}, \dots, \just{\pv[n]}{A_n}]{\just{t}{A}}
    $$
\end{proposition}
More importantly, we have shown here an embedding of proofs of $\IPL$ into~$\J$.

\begin{remark}
    $\J$ also allows for $\alpha$-renaming like the $\lambda$-calculus.
    Additionally, the substitution property of propositional atoms also holds, that is if $\hprove{A}$, then $\hprove{A[P = B]}$ for any $B \in \langjust$ for occurrences of $P \in \Prop$ in $A$.
    This is in contrast to~$\LP$ and standard justification logics in both classical and intuitionistic settings due to justification logic normally being \emph{hyperintensional} (see~\cite{artemov_justification_2024}).
\end{remark}

Unlike the hypothetical logic of proofs~\cite{steren_intuitionistic_2014} and the history based style computations of justification logic~\cite{bavera_justification_2010,bavera_justification_2018,barenbaum_rewrites_2020}, we have the following as a theorem of~$\J$.
\begin{proposition}\label{prop:typeformula}
    $\hprove{A \IMP \just{x}{A}}$ for $x \in \prfvar$.
\end{proposition}
\begin{proof}
    We begin with the axioms 
    $\hprove{\just{x}{A} \IMP \just{\pbang{x}}{\just{x}{A}}}$ 
    and 
    $\hprove{(\just{x}{A} \IMP \just{\pbang{x}}{\just{x}{A}}) \IMP \just{\pl{\pbang{x}}}{(A \IMP \just{x}{A})}}$.
    Applying $\mp$ we get
    $\hprove{\just{\pl{\pbang{x}}}{(A \IMP \just{x}{A})}}$.
    Apply $\mp$ with the following instance of the $\jfax$ axiom 
    $\hprove{\just{\pl{\pbang{x}}}{(A \IMP \just{x}{A})} \IMP (A \IMP \just{x}{A})}$ 
    to get
    $\hprove{A \IMP \just{x}{A}}$.
\end{proof}
This theorem can be simply interpreted that any well-formed formula of the language can be typed by a variable -- a property unavailable to systems that track history~\cite{bavera_justification_2010,steren_intuitionistic_2014,bavera_justification_2018,barenbaum_rewrites_2020}, since the identity of the variable there is tied to how the proof was built.

\section{Natural Deduction for $\J$}\label{sec:ND}
In this section we introduce the untyped~$\NDJ$ and typed~$\laJ$ natural deduction calculi.
We explore the local reductions and explore how the term calculus relates to the proof terms~$\prfterms$ of the logic. 

\subsection{Untyped Natural Deduction Calculus $\NDJ$}
An untyped judgement is defined as follows.
\begin{definition}
	An~\emph{untyped judgement} is written $\utctx{\Gamma}{A}$, where $\Gamma$ is an (unordered) multiset of formulas in $\langjust$ and $A \in \langjust$ is defined inductively by
	the natural deduction calculus~$\NDJ$ given by the rules in \cref{fig:jlnd} with~$\id$ as the base case, and a \emph{proof} in $\NDJ$ is constructed as trees from these rules with root $\utctx{\Gamma}{A}$ and leaves closed by $\id$.
\end{definition}

\begin{figure}[t]
    \fbox
    {
        \begin{minipage}{.95\textwidth}
            \centering
            $
                \vlinf{\id}{}
                {
                    \utctx{\Gamma, A}{A}
                }
                {}
            \qquad
                \vlinf{\intrle{\IMP}}{}
                {
                    \utctx{\Gamma}{A \IMP B}
                }
                {
                    \utctx{\Gamma, A}{B}
                }
                \qquad
                \vliinf{\elimrle{\IMP}}{}
                {
                    \utctx{\Gamma}{B}
                }
                {
                    \utctx{\Gamma}{A \IMP B}
                }
                {
                    \utctx{\Gamma}{A}
                }
            $
            \\[1ex]
            $
                \vliinf{\intrle{\AND}}{}
                {
                    \utctx{\Gamma}{A \AND B}
                }
                {
                    \utctx{\Gamma}{A}
                }
                {
                    \utctx{\Gamma}{B}
                }
                \qquad
                \vlinf{\elimrleleft{\AND}}{}
                {
                    \utctx{\Gamma}{A}
                }
                {
                    \utctx{\Gamma}{A \AND B}
                }
                \qquad
                \vlinf{\elimrleright{\AND}}{}
                {
                    \utctx{\Gamma}{B}
                }
                {
                    \utctx{\Gamma}{A \AND B}
                }
            $
            \\[1ex]
            $
                \vlinf{\intrle{\boxrule[\ ]}}{}
                {
                    \utctx{\Gamma}{\just{x}{A}}
                }
                {
                    \utctx{\Gamma}{A}
                }
                \qquad
                \vlinf{\intrle{\boxrule[\pbang{} ]}}{}
                {
                    \utctx{\Gamma}{\just{\pbang{t}}{\just{t}{A}}}
                }
                {
                    \utctx{\Gamma}{\just{t}{A}}
                }
                \qquad
                \vlinf{\elimrle{\boxrule[\ ]}}{}
                {
                    \utctx{\Gamma}{A}
                }
                {
                    \utctx{\Gamma}{\just{t}{A}}
                }
            $
            \\[1ex]
            $
                \vliinf{\intrle{\boxrule[\pdot{}{}]}}{}
                {
                    \utctx{\Gamma}{\just{\pdot{s}{t}}{B}}
                }
                {
                    \utctx{\Gamma}{\just{s}{(A \IMP B)}}
                }
                {
                    \utctx{\Gamma}{\just{t}{A}}
                }
                \qquad
                \vlinf{\intrle{\boxrule[\pl{}]}}{}
                {
                    \utctx{\Gamma}{\just{\pl{t}}{(A \IMP B)}}
                }
                {
                    \utctx{\Gamma, \just{x}{A}}{\just{t}{B }}
                }
            $
            \\[1ex]
            $
                \vlinf{\intrle{\boxrule[\pael{}]}}{}
                {
                    \utctx{\Gamma}{\just{\pael{t}}{A}}
                }
                {
                    \utctx{\Gamma}{\just{t}{(A \AND B)}}
                }
                \quad
                \vlinf{\intrle{\boxrule[\paer{}]}}{}
                {
                    \utctx{\Gamma}{\just{\paer{t}}{B}}
                }
                {
                    \utctx{\Gamma}{\just{t}{(A \AND B)}}
                }
            $
            \\[1ex]
            $
                \vliinf{\intrle{\boxrule[\pai{}{}]}}{}
                {
                    \utctx{\Gamma}{\just{\pai{s}{t}}{(A \AND B)}}
                }
                {
                    \utctx{\Gamma}{\just{s}{A}}
                }
                {
                    \utctx{\Gamma}{\just{t}{B}}
                }
            $
        \end{minipage}
    }
    \caption{Natural deduction calculus $\NDJ$}
    \label{fig:jlnd}
\end{figure}

\begin{theorem}[Equivalence $\J$ and $\NDJ$]\label{thm:ndsoundcomp}
	$\hprove[\Gamma]{A} \iff \utctx{\Gamma}{A}$
\end{theorem}
\begin{proof}
	The~$\Leftarrow$ direction proceeds by induction on~$\utctx{\Gamma}{A}$ and showing each rule
	It can be seen each rule is locally sound by usual axiomatic reasoning and utilising \cref{prop:typeformula} for $\intrle{\boxrule[\ ]}$.
	
	The~$\Rightarrow$ direction follows by induction on $\hprove[\Gamma]{A}$.
	In the case $A \in \Gamma$, we have
	$$
	\vlinf{\id}{}
	{
		\utctx{\Gamma}{A}
	}
	{}.
	$$
	In the case~$A$ is an axiom instance, we give a proof of them in the natural deduction calculus.
	We prove a justification axiom here:
	$$
	\vlderivation
	{
		\vlin{\intrle{\IMP}}{}
		{
			\utctx{}{\just{s}{(A \IMP B)} \IMP \just{t}{A} \IMP \just{\pdot{s}{t}}{B}}
		}
		{
			\vlin{\intrle{\IMP}}{}
			{
				\utctx{\just{s}{(A \IMP B)}}{\just{t}{A} \IMP \just{\pdot{s}{t}}{B}}
			}
			{
				\vliin{\intrle{\boxrule[\pdot{}{}]}}{}
				{
					\utctx{\just{s}{(A \IMP B)}, \just{t}{A}}{ \just{\pdot{s}{t}}{B}}
				}
				{
					\vlin{\id}{}
					{
						\utctx{\just{s}{(A \IMP B)}, \just{t}{A}}{ \just{s}{(A \IMP B)}}
					}
					{
						\vlhy{}
					}
				}
				{
					\vlin{\id}{}
					{
						\utctx{\just{s}{(A \IMP B)}, \just{t}{A}}{ \just{t}{B}}
					}
					{
						\vlhy{}
					}
				}
			}
		}
	}
	$$

	\noindent
	For the other justification axioms, see \cref{ex:ndjustax} in the Appendix.
	And $\mp$ is simulated with~$\elimrle{\IMP}$ as expected.
\end{proof}

\subsection{Typed Natural Deduction Calculus $\laJ$}
The term calculus is given by the following.
\begin{definition}
	The term calculus, denoted as the set~$\terms$ is defined by the following grammar:
	\begin{multline*}
		M
		\BNFequals
		\var
		\mid
		\app{M}{M}
		\mid
		\la{M}
		\mid
		\projl{M}
		\mid
		\projr{M}
		\mid
		\andc{M}{M}
		\\
		\mid
		\prom{\pv}{M}
		\mid
		\bang{M}
		\mid
		\boxapp{M}{M}
		\mid
		\boxla{M}
		\mid
		\boxprojl{M}
		\mid
		\boxprojr{M}
		\mid
		\boxandc{M}{M}
		\mid
		\unbox{M}
	\end{multline*}
	where $\var$ ranges over a countable set of variables~$\Vars$.
\end{definition}
The term calculus can be seen as an extension of $\IPLterms$ by taking into account the additional justification operations.
We now define the following.
\begin{definition}
	A~\emph{typed judgement} is written $\ctx{\typset{\Gamma}}{\fm{M}{A}}$, where $\typset{\Gamma}$ is a set of typed formulas in $\langjust$, typed by \emph{distinct} variables and $A \in \langjust$ is defined inductively by
	the natural deduction calculus~$\laJ$ given by the rules in \cref{fig:typednd} with~$\id$ as the base case, and a \emph{proof} in $\laJ$ is constructed as trees from these rules with root $\ctx{\typset{\Gamma}}{\fm{M}{A}}$ and leaves closed by $\id$.
\end{definition}

\begin{figure}[t]
    \fbox
    {
        \begin{minipage}{.95\textwidth}
            \centering
            $
                \vlinf{\id}{}
                {
                    \ctx{\typset{\Gamma}, \fm{\var}{A}}{\fm{\var}{A}}
                }
                {}
            $
            \\[1ex]
            $
                \vlinf{\intrle{\IMP}}{}
                {
                    \ctx{\typset{\Gamma}}{\fm{\la{M}}{A \IMP B}}
                }
                {
                    \ctx{\typset{\Gamma}, \fm{\var}{A}}{\fm{M}{B}}
                }
                \qquad
                \vliinf{\elimrle{\IMP}}{}
                {
                    \ctx{\typset{\Gamma}}{\fm{\app M N}{B}}
                }
                {
                    \ctx{\typset{\Gamma}}{\fm{M}{A \IMP B}}
                }
                {
                    \ctx{\typset{\Gamma}}{\fm N A}
                }
            $
            \\[1ex]
            $
                \vliinf{\intrle{\AND}}{}
                {
                    \ctx{\typset{\Gamma}}{\fm{\andc M N}{A \AND B}}
                }
                {
                    \ctx{\typset{\Gamma}}{\fm{M}{A}}
                }
                {
                    \ctx{\typset{\Gamma}}{\fm{N}{B}}
                }
            $
            \\[1ex]
            $
                \vlinf{\elimrleleft{\AND}}{}
                {
                    \ctx{\typset{\Gamma}}{\fm{\projl M}{A}}
                }
                {
                    \ctx{\typset{\Gamma}}{\fm{M}{A \AND B}}
                }
                \qquad
                \vlinf{\elimrleright{\AND}}{}
                {
                    \ctx{\typset{\Gamma}}{\fm{\projr M}{B}}
                }
                {
                    \ctx{\typset{\Gamma}}{\fm{M}{A \AND B}}
                }
            $
            \\[1ex]
            $
                \vlinf{\intrle{\boxrule[\ ]}}{}
                {
                    \ctx{\typset{\Gamma}}{\fm{\prom{x}{M}}{\just{x}{A}}}
                }
                {
                    \ctx{\typset{\Gamma}}{\fm M A}
                }
                \qquad
                \vlinf{\intrle{\boxrule[\pbang{}]}}{}
                {
                    \ctx{\typset{\Gamma}}{\fm{\bang{M}}{\just{\pbang{t}}{\just{t}{A}}}}
                }
                {
                    \ctx{\typset{\Gamma}}{\fm{M}{\just{t}{A}}}
                }
                \qquad
                \vlinf{\elimrle{\boxrule[\ ]}}{}
                {
                    \ctx{\typset{\Gamma}}{\fm{\unbox{M}}{A}}
                }
                {
                    \ctx{\typset{\Gamma}}{\fm{M}{\just{t}{A}}}
                }
            $
            \\[1ex]
            $
                \vliinf{\intrle{\boxrule[\pdot{}{}]}}{}
                {
                    \ctx{\typset{\Gamma}}{\fm{\boxapp M N}{\just{\pdot{s}{t}}{B}}}
                }
                {
                    \ctx{\typset{\Gamma}}{\fm{M}{\just{s}{(A \IMP B)}}}
                }
                {
                    \ctx{\typset{\Gamma}}{\fm{N}{\just{t}{A}}}
                }
            $
            \\[1ex]
            $
                \vlinf{\intrle{\boxrule[\pl{}]}}{}
                {
                    \ctx{\typset{\Gamma}}{\fm{\boxla M}{\just{\pl{t}}{(A \IMP B)}}}
                }
                {
                    \ctx{\typset{\Gamma}, \fm{\var}{\just{x}{A}}}{\fm{M}{\just{t}{B}}}
                }
            $
            \\[1ex]
            $
                \vlinf{\intrle{\boxrule[\pael{}]}}{}
                {
                    \ctx{\typset{\Gamma}}{\fm{\boxprojl M}{\just{\pael{t}}{A}}}
                }
                {
                    \ctx{\typset{\Gamma}}{\fm{M}{\just{t}{(A \AND B)}}}
                }
                \quad
                \vlinf{\intrle{\boxrule[\paer{}]}}{}
                {
                    \ctx{\typset{\Gamma}}{\fm{\boxprojr M}{\just{\paer{t}}{B}}}
                }
                {
                    \ctx{\typset{\Gamma}}{\fm{M}{\just{t}{(A \AND B)}}}
                }
                $
                \\[1ex]
                $
                \vliinf{\intrle{\boxrule[\pai{}{}]}}{}
                {
                    \ctx{\typset{\Gamma}}{\fm{\boxandc M N}{\just{\pai{s}{t}}{(A \AND B)}}}
                }
                {
                    \ctx{\typset{\Gamma}}{\fm{M}{\just{s}{A}}}
                }
                {
                    \ctx{\typset{\Gamma}}{\fm{N}{\just{t}{B}}}
                }
            $
        \end{minipage}
    }
    \caption{Natural deduction calculus~$\laJ$}
    \label{fig:typednd}
\end{figure}

Now we can state the Curry-Howard correspondence.

\begin{theorem}\label{prop:jlcurryhoward}
	For multiset $\Gamma \sset \langpl$ and $A \in \langpl$.
	We write $\typset{\Gamma}$ the set typing the formulas $\Gamma$ with distinct proof variables.
	Then for some term $M$ constructed from proof variables in $\Gamma$
	$$
	\utctx{\Gamma}{A}
	\iff
	\ctx{\typset{\Gamma}}{\fm{M}{A}}
	$$
\end{theorem}

\begin{proof}
	The~$\Rightarrow$ direction is a routine proof by induction on~$\utctx{\Gamma}{A}$.
	The~$\Leftarrow$ can be seen from the fact that any derivation of~$\laJ$ is a derivation in~$\NDJ$ by translating each formula~$\fm{M}{A}$ contained in the judgement to~$A$.
\end{proof}

The following is an example of typing a proof of the $\jimpi$~axiom.
\begin{example}
	$$
	\vlderivation
	{
		\vlin{\intrle{\IMP}}{}
		{
			\ctx{}{\fm{\la[\var[1]]{\boxla[\var[2]]{\app{\var[1]}{\var[2]}}}}{(\just{x}{A} \IMP \just{t}{B}) \IMP \just{\pl{t}}{(A \IMP B)}}}
		}
		{
			\vlin{\intrle{\boxrule[\pl{}]}}{}
			{
				\ctx{\fm{\var[1]}{\just{x}{A} \IMP \just{t}{B}}}{\fm{\boxla[\var[2]]{\app{\var[1]}{\var[2]}}}{\just{\pl{t}}{(A \IMP B)}}}
			}
			{
				\vliin{\elimrle{\IMP}}{}
				{
					\ctx{\fm{\var[1]}{\just{x}{A} \IMP \just{t}{B}}, \fm{\var[2]}{\just{x}{A}}}{\fm{\app{\var[1]}{\var[2]}}{\just{t}{B}}}
				}
				{
					\vlin{\id}{}
					{
						\ctx{\fm{\var[1]}{\just{x}{A} \IMP \just{t}{B}}, \fm{\var[2]}{\just{x}{A}}}{\fm{\var[1]}{\just{x}{A} \IMP \just{t}{B}}}
					}
					{
						\vlhy{}
					}
				}
				{
					\vlin{\id}{}
					{
						\ctx{\fm{\var[1]}{\just{x}{A} \IMP \just{t}{B}}, \fm{\var[2]}{\just{x}{A}}}{\fm{\var[2]}{\just{x}{A}}}
					}
					{
						\vlhy{}
					}
				}
			}
		}
	}
	$$
\end{example}
For the other justification axioms, see \cref{ex:typjustax} in the Appendix.

Now we will discuss \emph{normalisation}.
We define detours in the usual way.
\begin{definition}
	A proof in $\mathsf{ND}\J$ or $\lambda \J$ contains a \emph{detour} if there is a branch which contains an introduction rule, followed immediately by an elimination rule.
	
	When there is a normal proof of $\utctx{\Gamma}{A}$ (respectively $\ctx{\typset{\Gamma}}{\fm{M}{A}}$) we write $\nutctx{\Gamma}{A}$ (respectively $\nctx{\typset{\Gamma}}{\fm{M}{A}}$).
\end{definition}
Before we define our reductions, we define substitution of terms in the routine way.
\begin{definition}[Substitution of terms]
	Given a term $M \in \terms$, $\subst{M}{\var}{N}$ is defined inductively:
	\begin{itemize}
		\item $\subst{a}{\var}{N} \colonequals N$;
		\item $\subst{b}{\var}{N} \colonequals b$ when $b \neq \var$;
		\item $\subst{\la{M}}{\var}{N} \colonequals \la{M}$;
		\item $\subst{\la[b]{M}}{\var}{N} \colonequals \la[b]{\subst{M}{\var}{N}}$ where $b \neq \var$;
		\item $\subst{\boxla{M}}{\var}{N} \colonequals \boxla{M}$;
		\item $\subst{\boxla[b]{M}}{\var}{N} \colonequals \boxla[b]{\subst{M}{\var}{N}}$ where $b \neq \var$;
		\item For unary constructors $\mathsf{C}$, $\subst{\mathsf{C}(M)}{\var}{N} \colonequals \mathsf{C}(\subst{M}{\var}{N})$;
		\item For the remaining binary constructors $\mathsf{D}$, $\subst{\mathsf{D}(M_1, M_2)}{\var}{N} \colonequals \mathsf{D}(\subst{M_1}{\var}{N},\subst{M_2}{\var}{N})$.
	\end{itemize}
\end{definition}
Similarly to $\lambda$, we also have $\lambda^{\BOX}$ \emph{binding} variables.
We will utilise the following.
\begin{lemma}
	Given $\ctx{\typset{\Gamma}, \fm{\var}{\just{x}{A}}}{\fm{M}{B}}$, we have $\ctx{\typset{\Gamma}, \fm{\var'}{A}}{\fm{\subst{M}{\var}{\prom{x}{\var'}}}{B}}$.
\end{lemma}
\begin{proof}
	This is a routine proof by induction on~$\ctx{\typset{\Gamma}, \fm{\var}{\just{x}{A}}}{\fm{M}{B}}$.
\end{proof}
It is worth noting that the translation above could introduce new detours.

Now we can illustrate our reductions (we omit the reductions of~$\NDIPL$):
{$$
	\vlderivation
	{
		\vlin{\elimrle{\boxrule[\ ]}}{}
		{
			\ctx{\typset{\Gamma}}{\ctxfm{\unbox{\prom{x}{M}}}{A}}
		}
		{
			\vlin{\intrle{\boxrule[\ ]}}{}
			{
				\ctx{\typset{\Gamma}}{\ctxfm{\prom{x}{M}}{\just{x}{A}}}
			}
			{
				\vlhy{\ctx{\typset{\Gamma}}{\ctxfm{M}{A}}}
			}
		}
	}
	\quad
	\rightsquigarrow
	\quad
	\ctx{\typset{\Gamma}}{\ctxfm{M}{A}}
	$$
	
	$$
	\vlderivation
	{
		\vlin{\elimrle{\boxrule[\ ]}}{}
		{
			\ctx{\typset{\Gamma}}{\ctxfm{\unbox{\bang{M}}}{\just{t}{A}}}
		}
		{
			\vlin{\intrle{\boxrule[\pbang{} ]}}{}
			{
				\ctx{\typset{\Gamma}}{\ctxfm{\bang{M}}{\just{\pbang{t}}{\just{t}{A}}}}
			}
			{
				\vlhy{\ctx{\typset{\Gamma}}{\ctxfm{M}{\just{t}{A}}}}
			}
		}
	}
	\quad
	\rightsquigarrow
	\quad
	\ctx{\typset{\Gamma}}{\ctxfm{M}{\just t A}}
	$$
	
	$$
	\vlderivation
	{
		\vlin{\elimrle{\boxrule[\ ]}}{}
		{
			\ctx{\typset{\Gamma}}{\fm{\unbox{\boxapp M N}}{B}}
		}
		{
			\vliin{\intrle{\boxrule[\pdot{}{}]}}{}
			{
				\ctx{\typset{\Gamma}}{\fm{\boxapp M N}{\just{\pdot{s}{t}}{B}}}
			}
			{
				\vlhy{\ctx{\typset{\Gamma}}{\fm{M}{\just{s}{(A \IMP B)}}}}
			}
			{
				\vlhy{\ctx{\typset{\Gamma}}{\fm{N}{\just{t}{A}}}}
			}
		}
	}
	\quad
	\rightsquigarrow
	\quad
	\vlderivation
	{
		\vliin{\elimrle{\IMP}}{}
		{
			\ctx{\typset{\Gamma}}{\fm{\app{\unbox{M}}{\unbox{N}}}{B}}
		}
		{
			\vlin{\elimrle{\boxrule[\ ]}}{}
			{
				\ctx{\typset{\Gamma}}{\fm{\unbox{M}}{A \IMP B}}
			}
			{
				\vlhy{\ctx{\typset{\Gamma}}{\fm{M}{\just{s}{(A \IMP B)}}}}
			}
		}
		{
			\vlin{\elimrle{\boxrule[\ ]}}{}
			{
				\ctx{\typset{\Gamma}}{\fm{\unbox{N}}{A}}
			}
			{
				\vlhy{\ctx{\typset{\Gamma}}{\fm{N}{\just{t}{A}}}}
			}
		}
	}
	$$
	
	$$
	\vlderivation
	{
		\vlin{\elimrle{\boxrule[\ ]}}{}
		{
			\ctx{\typset{\Gamma}}{\fm{\unbox{\boxla M}}{A \IMP B}}
		}
		{
			\vlin{\intrle{\boxrule[\pl{}]}}{}
			{
				\ctx{\typset{\Gamma}}{\fm{\boxla M}{\just{\pl{t}}{(A \IMP B)}}}
			}
			{
				\vlhy{\ctx{\typset{\Gamma}, \fm{\var}{\just{x}{A}}}{\fm{M}{\just{t}{A}}}}
			}
		}
	}
	\quad
	\rightsquigarrow
	\quad
	\vlderivation
	{
		\vlin{\intrle{\IMP}}{}
		{
			\ctx{\typset{\Gamma}}{\fm{\la[\var']{\unbox{\subst{M}{\var}{\prom{x}{\var'}}}}}{A \IMP B}}
		}
		{
			\vlin{\elimrle{\boxrule[\ ]}}{}
			{
				\ctx{\typset{\Gamma}, \fm{\var'}{A}}{\fm{\unbox{\subst{M}{\var}{\prom{x}{\var'}}}}{B}}
			}
			{
				\vlhy{\ctx{\typset{\Gamma}, \fm{\var'}{A}}{\fm{\subst{M}{\var}{\prom{x}{\var'}}}{\just{t}{B}}}}
			}
		}
	}
	$$
	
	$$
	\vlderivation
	{
		\vlin{\elimrle{\boxrule[\ ]}}{}
		{
			\ctx{\typset{\Gamma}}{\fm{\unbox{\boxprojl M}}{A}}
		}
		{
			\vlin{\intrle{\boxrule[\pael{}]}}{}
			{
				\ctx{\typset{\Gamma}}{\fm{\boxprojl M}{\just{\pael{t}}{A}}}
			}
			{
				\vlhy{\ctx{\typset{\Gamma}}{\fm{M}{\just{t}{(A \AND B)}}}}
			}
		}
	}
	\quad
	\rightsquigarrow
	\quad
	\vlderivation
	{
		\vlin{\elimrleleft{\AND}}{}
		{
			\ctx{\typset{\Gamma}}{\fm{\projl{\unbox{M}}}{A}}
		}
		{
			\vlin{\elimrle{\boxrule[\ ]}}{}
			{
				\ctx{\typset{\Gamma}}{\fm{\unbox{M}}{A \AND B}}
			}
			{
				\vlhy{\ctx{\typset{\Gamma}}{\fm{M}{\just{t}{(A \AND B)}}}}
			}
		}
	}
	$$
	
	$$
	\vlderivation
	{
		\vlin{\elimrle{\boxrule[\ ]}}{}
		{
			\ctx{\typset{\Gamma}}{\fm{\unbox{\boxprojr M}}{{B}}}
		}
		{
			\vlin{\intrle{\boxrule[\paer{}]}}{}
			{
				\ctx{\typset{\Gamma}}{\fm{\boxprojr M}{\just{\paer{t}}{B}}}
			}
			{
				\vlhy{\ctx{\typset{\Gamma}}{\fm{M}{\just{t}{(A \AND B)}}}}
			}
		}
	}
	\quad
	\rightsquigarrow
	\quad
	\vlderivation
	{
		\vlin{\elimrleright{\AND}}{}
		{
			\ctx{\typset{\Gamma}}{\fm{\projr{\unbox{M}}}{B}}
		}
		{
			\vlin{\elimrle{\boxrule[\ ]}}{}
			{
				\ctx{\typset{\Gamma}}{\fm{\unbox{M}}{A \AND B}}
			}
			{
				\vlhy{\ctx{\typset{\Gamma}}{\fm{M}{\just{t}{(A \AND B)}}}}
			}
		}
	}
	$$
	
	$$
	\vlderivation
	{
		\vlin{\elimrle{\boxrule[\ ]}}{}
		{
			\ctx{\typset{\Gamma}}{\fm{\unbox{\boxandc M N}}{A \AND B}}
		}
		{
			\vliin{\intrle{\boxrule[\pai{}{}]}}{}
			{
				\ctx{\typset{\Gamma}}{\fm{\boxandc M N}{\just{\pai{s}{t}}{(A \AND B)}}}
			}
			{
				\vlhy{\ctx{\typset{\Gamma}}{\fm{M}{\just{s}{A}}}}
			}
			{
				\vlhy{\ctx{\typset{\Gamma}}{\fm{N}{\just{t}{B}}}}
			}
		}
	}
	\quad
	\rightsquigarrow
	\quad
	\vlderivation
	{
		\vliin{\intrle{\AND}}{}
		{
			\ctx{\typset{\Gamma}}{\fm{\andc{\unbox{M}}{\unbox{N}}}{A \AND B}}
		}
		{
			\vlin{\elimrle{\boxrule[\ ]}}{}
			{
				\ctx{\typset{\Gamma}}{\fm{\unbox{M}}{A}}
			}
			{
				\vlhy{\ctx{\typset{\Gamma}}{\fm{M}{\just{s}{A}}}}
			}
		}
		{
			\vlin{\elimrle{\boxrule[\ ]}}{}
			{
				\ctx{\typset{\Gamma}}{\fm{\unbox{N}}{B}}
			}
			{
				\vlhy{\ctx{\typset{\Gamma}}{\fm{N}{\just{t}{A}}}}
			}
		}
	}
	$$}
\normalsize

\noindent
As can be seen in these reductions, additional detours that can be introduced in applying one of these reductions to a proof could be of higher \emph{degree}, so a different measure will be required to achieve \emph{normalisation}.
Doing this directly in the natural deduction calculus will be difficult so we will achieve this instead through a \emph{cut-elimination} of a Gentzen-style formulation of $\J$ (\cref{sec:pt}).

\subsection{Relating to Self-internalisation}
In this subsection, we are going to relate $\prfterms$, the syntactic proof terms of $\J$ and $\terms$, the computational proof terms obtained from the typed natural deduction of $\J$.
To achieve this, we require some bookkeeping on the variables of both $\prfterms$ and $\terms$:
note that~$\langjust$ is countable so enumerate each variable~$\pv$ and~$\var$ with formulas of the language,~i.e.
$$
\prfvar \colonequals \set[{\pv[A]}]{A \in \langjust}
\qquad
\Vars \colonequals \set[{\var[A]}]{A \in \langjust}
$$
This is in effect a \emph{Church}-style typing of variables.

\textbf{From $\prfterms$ to $\terms$.}
Here, we will define a map $\function{\ptttsymb}{\prfterms}{\terms}$ such that 
we can strengthen the Self-internalisation Theorem (\cref{thm:internalisation}):
\begin{theorem}\label{thm:inttond}
	If $\hprove{A}$, then there exists a closed proof term $t$ such that $\hprove{\just{t}{A}}$
	and $\ctx{}{\fm{\pttt{t}}{A}}$. 
\end{theorem}

We define a partial map in the following way.
\begin{definition}
	$\function{\ptttsymb}{\prfterms}{\terms}$
	is a partial map inductively defined by the following:
	\begin{itemize}
		\item If $t(\pv[A_1], \dots. \pv[A_n]) \in \IPLterms$, then $t(\pv[A_1], \dots. \pv[A_n]) \mapsto t(\var[A_1], \dots. \var[A_n])$;
		\item $\pv[A] \mapsto \var[A]$;
		\item $\pdot{s}{t} \mapsto \app{\ptttsymb(s)}{\ptttsymb(t)}$;
		\item $\pl[\pv[\just{s}{(A \IMP B)}]]{\pl[\pv[\just{t}{A}]] {\pbang{(\pdot{s}{t})}}} \mapsto \la[\var[\just{s}{(A \IMP B)}]]{\la[\var[\just{t}{A}]]{\boxapp{\var[\just{s}{(A \IMP B)}]}{\var[\just{t}{A}]}}}$;
		\item $\pl[\pv[\just{t}{A}]]{\pbang{\pbang{t}}} \mapsto \la[{\var[\just{t}{A}]}]{\bang{\var[\just{t}{A}]}}$;
		\item $\pl[\pv[\just{t}{A}]]{t} \mapsto \la[\var[\just{t}{A}]]{\unbox{\var[\just{t}{A}]}}$;
		\item $\pl[\pv[\just{x}{A} \IMP \just{t}{B}]]{\pbang{\pl{t}}} \mapsto \la[\var[\just{x}{A} \IMP \just{t}{B}]]{\boxla[\var[\just{x}{A}]]{\app{\var[\just{x}{A} \IMP \just{t}{B}]}{\var[\just{x}{A}]}}}$
		\item $\pl[\pv[\just{t}{(A \AND B)}]]{\pbang{\pael{t}}} \mapsto \la[\var[\just{t}{(A \AND B)}]]{\boxprojl{\var[\just{t}{(A \AND B)}]}}$;
		\item $\pl[\pv[\just{t}{(A \AND B)}]]{\pbang{\paer{t}}} \mapsto \la[\var[\just{t}{(A \AND B)}]]{\boxprojr{\var[\just{t}{(A \AND B)}]}}$;
		\item $\pl[\pv[\just{s}{A}]]{\pl[\pv[\just{t}{B}]]{\pbang{\pai{s}{t}}}} \mapsto \la[\var[\just{s}{A}]]{\la[\var[\just{t}{B}]]{\boxandc{\var[\just{s}{A}]}{\var[\just{t}{B}]}}}$.
	\end{itemize}
\end{definition}
The $\pbang{}$~operator in $\prfterms$ indicates if we are mapping the terms to a justification operation in~$\terms$.
As was discussed in \cref{rem:iplthms}, proof terms obtained for theorems of~$\IPL$ will be~$\IPLterms$ and its terms in $\laJ$ will be of the same shape but with different variable naming.

\begin{proof}[Proof of \cref{thm:inttond}]
	Repeat the proof of the Self-Internalisation Theorem (\cref{thm:internalisation}), but for each use of a proof variable being used for some formula~$B$, i.e.~$\just{\pv}{B}$, use $\just{\pv[B]}{B}$ instead to construct the proof term~$t$ and to track the necessary information (see \cref{ex:inttyped} in the Appendix for the justification axioms).
	Then, in the proof term obtained for $\utctx{}{\just{t}{A}}$ using \cref{thm:ndsoundcomp,prop:jlcurryhoward},~$M$, we have $\pttt{t} = M$.
\end{proof}

\textbf{From $\terms$ to $\prfterms$.}
Here, we will define a map $\function{\ttptsymb}{\terms}{\prfterms}$ such that we have the following.

\begin{theorem}\label{thm:tmtoptm}
	If $\ctx{\fm{\var[A_1]}{A_1}, \dots, \fm{\var[A_n]}{A_n}}{\fm{M(\var[A_1], \dots, \var[A_n])}{A}}$
	then we have
	$\hprove[\just{\ttpt{\var[A_1]}}{A_1}, \dots, \just{\ttpt{\var[A_n]}}{A_n}]{\just{\ttpt{M(\var[A_1], \dots, \var[A_n])}}{A}}$.
\end{theorem}

We will utilise the Church-style typing to identify the \emph{unique} type of the terms.

\begin{lemma}\label{lem:ttpt}
	For $\var[A_1], \dots, \var[A_n] \in \Vars$ and $M(\var[A_1], \dots, \var[A_n]) \in \terms$,
	then if $\ctx{\fm{\var[A_1]}{A_1}, \dots, \fm{\var[A_n]}{A_n}}{\fm{M(\var[A_1], \dots, \var[A_n])}{A}}$
	and
	$\ctx{\fm{\var[A_1]}{A_1}, \dots, \fm{\var[A_n]}{A_n}}{\fm{M(\var[A_1], \dots, \var[A_n])}{B}}$,
	then $A \equiv B$.
	
	If $A = \just{t}{A'}$ for some proof term $t$ and formula $A'$, denote $\boxterm{M(\var[A_1], \dots, \var[A_n])} \colonequals t$.
\end{lemma}
\begin{proof}
	This is a similar proof to the case of the $\lambda$-calculus of~$\IPL$~\cite{sorensen_lectures_2006}.
\end{proof}
Being able to trace what formula is being represented allows to define the following map.
\begin{definition}\label{def:ttpt}
	$\function{\ttptsymb}{\terms}{\prfterms}$
	is a partial map defined inductively
	\begin{itemize}
		\item $\ttpt{\var[A]} \colonequals \pv[A]$;
		
		\item $\ttpt{\app{M}{N}} \colonequals \pdot{\ttpt{M}}{\ttpt{N}}$
		
		\item $\ttpt{\la[\var[A]]{M}} \colonequals \pl[\pv[A]]{\ttpt{M}}$;
		
		\item $\ttpt{\projl{M}} \colonequals \pael{\ttpt{M}}$;
		
		\item $\ttpt{\projr{M}} \colonequals \paer{\ttpt{M}}$;
		
		\item $\ttpt{\andc{M}{N}} \colonequals \pai{\ttpt{M}}{\ttpt{N}}$;
		
		\item $\ttpt{\prom{\pv}{M}} \colonequals \pbang{\pv}$;
		
		\item $\ttpt{\bang{M(\var[A_1], \dots, \var[A_n])}} \colonequals \pbang{\pbang{\boxterm{M(\var[A_1], \dots, \var[A_n])}}}$;
		
		\item $\ttpt{\boxapp{M(\var[A_1], \dots, \var[A_n])}{N(\var[A_1], \dots, \var[A_n])}}$ 
		\vspace{-0.3cm}
		\begin{flushright}
			$ \colonequals \pbang{(\pdot{\boxterm{M(\var[A_1], \dots, \var[A_n])}}{\boxterm{N(\var[A_1], \dots, \var[A_n])}})}$;
		\end{flushright}

		\item $\ttpt{\boxla[\var[\just{\pv[A]}{A}]]{M(\var[A_1], \dots, \var[A_n], \var[\just{\pv[A]}{A}])}} \colonequals \pbang{\pl[\pv[A]]{\boxterm{M(\var[A_1], \dots, \var[A_n], \var[\just{\pv[A]}{A}])}}}$;
		
		\item $\ttpt{\boxprojl{M(\var[A_1], \dots, \var[A_n])}} \colonequals \bang{\pael{\boxterm{M(\var[A_1], \dots, \var[A_n])}}}$;
		
		\item $\ttpt{\boxprojr{M(\var[A_1], \dots, \var[A_n])}} \colonequals \bang{\paer{\boxterm{M(\var[A_1], \dots, \var[A_n])}}}$;
		
		\item ${\boxandc{M(\var[{A_1}], \dots, \var[{A_n}])}{N(\var[{A_1}], \dots, \var[{A_n}])}}^{\ttptsymb}$ 
		\vspace{-0.3cm}
		\begin{flushright}
			$\colonequals \pbang{\pai{\boxterm{M(\var[A_1], \dots, \var[A_n])}}{\boxterm{N(\var[A_1], \dots, \var[A_n])}}}$;
		\end{flushright}

		\item $\ttpt{\unbox{M(\var[A_1], \dots, \var[A_n])}} \colonequals \boxterm{M(\var[A_1], \dots, \var[A_n])}$.
	\end{itemize}
\end{definition}
The well-definedness of~$\ttptsymb$ is guaranteed through \cref{lem:ttpt}.

\begin{proof}[Proof of~\cref{thm:tmtoptm}]
	This proceeds by induction on~$\ctx{\fm{\var[A_1]}{A_1}, \dots, \fm{\var[A_n]}{A_n}}{\fm{M(\var[A_1], \dots, \var[A_n])}{A}}$
	unfolding \cref{def:ttpt} at each rule. 
	The soundness of each resulting instance follows routinely (for some cases, see the \cref{sec:appnd} in the Appendix).
\end{proof}

\section{Proof Theory}\label{sec:pt}
In this section, we explore Gentzen-style proof theory for $\J$.
Ultimately, it will be utilised for the purposes of normalisation
however
is also interesting to explore due to the features of the sequent calculus which we shall discuss, as done for lambda calculus in~\cite{BarendregtG00}.
To begin we define the following.
\begin{definition}
    A sequent is a pair $\seq{\Gamma}{A}$ where $\Gamma$ is an (unordered) \emph{multiset} of formulas in $\langjust$ and $A$ is a formula in $\langjust$.
    The \emph{interpretation} of $\seq{\Gamma}{A}$ is the formula
    $$
        \form{\seq{\Gamma}{A}} \colonequals \BIGAND \Gamma \IMP A
    $$
\end{definition}
%

\begin{definition}
    The sequent calculus~$\LJ$ is given by the rules in \cref{fig:LJ}, and the sequent calculus $\LJ + \cut$ is the system which includes the cut rule 
     $$
                    \vliinf{\cut}{}
                    {
    	                    \seq{\Gamma, \Delta}{B}
    	                }
                    {
    	                    \seq{\Gamma}{A}
    	                }
                    {
    	                    \seq{\Delta, A, \dots, A}{B}
    	                }
     $$
    A \emph{proof} in $\LJ$ ($\LJ + \cut$ respectively) of a sequent $\seq{\Gamma}{A}$ is constructed as trees from these rules with root $\seq{\Gamma}{A}$ and leaves closed by axiomatic rules;
    we write $\seqprove{\seq{\Gamma}{A}}$ ($\seqprovecut{\seq{\Gamma}{A}}$ respectively).
\end{definition}

One of the objectives of this section is to prove the following.
\begin{theorem}
    The following are equivalent:
    \begin{enumerate}
        \item $\hprove{A}$;
        \item $\seqprovecut{\seq{}{A}}$;
        \item $\seqprove{\seq{}{A}}$.
    \end{enumerate}
\end{theorem}
The direction 1~$\implies$~2 is a standard completeness argument shown in \cref{prop:seqcomplete}, 2~$\implies$~3 will follow from a cut-elimination argument (\cref{thm:cutelim}) and 3~$\implies$~1 follows by showing soundness of the rules in $\LJ$ (\cref{prop:soundness}).

As can be seen in \cref{fig:LJ}, the rules of~$\LJ$ are not \emph{analytic}, i.e.~some of the rules, namely~$\rr {\boxrule[\pdot{}{}]}$, $\rr{\boxrule[\pael{}]}$ and~$\rr{\boxrule[\paer{}]}$, do not satisfy the \emph{sub-formula} property.
Due to this, we are required to use a more intricate measure of a cut-formula to achieve cut-elimination.

\begin{figure}[t]
    \fbox
    {
        \begin{minipage}{.95\textwidth}
        \centering
            $
                \vlinf{\id}{}
                {
                    \seq{\Gamma, p}{p}
                }
                {}
                \qquad
                \vlinf{\id}{}
                {
                    \seq{\Gamma, \just{t}{A}}{\just{t}{A}}
                }
                {}
                \qquad
                \vlinf{\cont}{}
                {
                    \seq{\Gamma, A}{B}
                }
                {
                    \seq{\Gamma, A, A}{B}
                }
            $
            \\[1ex]
            $
                \vlinf{\lr \AND}{}
                {
                    \seq{\Gamma, A \AND B}{C}
                }
                {
                    \seq{\Gamma, A, B}{C}
                }
                \qquad
                \vliinf{\rr \AND}{}
                {
                    \seq{\Gamma}{A \AND B}
                }
                {
                    \seq{\Gamma}{A}
                }
                {
                    \seq{\Gamma}{B}
                }
           $
           \\[1ex]
           $
                \vliinf{\lr \IMP}{}
                {
                    \seq{\Gamma, A \IMP B}{C}
                }
                {
                    \seq{\Gamma, B}{C}
                }
                {
                    \seq{\Gamma}{A}
                }
                \qquad
                \vlinf{\rr \IMP}{}
                {
                    \seq{\Gamma}{A \IMP B}
                }
                {
                    \seq{\Gamma, A}{B}
                }
            $
            \\[1ex]
            $
                \vlinf{\lr {\boxrule[\ ]}}{}
                {
                    \seq{\Gamma, \just{t}{A}}{B}
                }
                {
                    \seq{\Gamma, A}{B}
                }
                \qquad
                \vlinf{\rr {\boxrule[\ ]}}{}
                {
                    \seq{\Gamma}{\just{x}{A}}
                }
                {
                    \seq{\Gamma}{A}
                }
                \qquad
                \vlinf{\rr {\boxrule[\pbang{}]}}{}
                {
                    \seq{\Gamma}{\just{\pbang{t}}{\just{t}{A}}}
                }
                {
                    \seq{\Gamma}{\just{t}{A}}
                }
            $
            \\[1ex]
            $
                \vliinf{\rr {\boxrule[\pdot{}{}]}}{}
                {
                    \seq{\Gamma}{\just{\pdot{s}{t}}{B}}
                }
                {
                    \seq{\Gamma}{\just{s}{(A \IMP B)}}
                }
                {
                    \seq{\Gamma}{\just{t}{A}}
                }
                \qquad
                \vlinf{\rr{\boxrule[\pl{}]}}{}
                {
                    \seq{\Gamma}{\just{\pl{t}}{(A \IMP B)}}
                }
                {
                    \seq{\Gamma, \just{x}{A}}{\just{t}{B}}
                }
            $
            \\[1ex]
            $
                \vlinf{\rr{\boxrule[\pael{}]}}{}
                {
                    \seq{\Gamma}{\just{\pael{t}}{A}}
                }
                {
                    \seq{\Gamma}{\just{t}{(A \AND B)}}
                }
                \quad
                \vlinf{\rr{\boxrule[\paer{}]}}{}
                {
                    \seq{\Gamma}{\just{\paer{t}}{B}}
                }
                {
                    \seq{\Gamma}{\just{t}{(A \AND B)}}
                }
                \qquad
                \vliinf{\rr{\boxrule[\pai{}{}]}}{}
                {
                    \seq{\Gamma}{\just{\pai{s}{t}}{(A \AND B)}}
                }
                {
                    \seq{\Gamma}{\just{s}{A}}
                }
                {
                    \seq{\Gamma}{\just{t}{B}}
                }
            $
        \end{minipage}
    }
    \caption{Sequent calculus $\LJ$}
    \label{fig:LJ}
    
\end{figure}

\subsection{Soundness and Completeness with respect to $\J$}
\begin{proposition}\label{prop:seqcomplete}
    $\hprove{A} \implies \seqprovecut{\seq{}{A}}$. 
\end{proposition}
\begin{proof}
    Proceed by induction on $\hprove{A}$.

    For the base cases, we show a proof of the axioms.
    For the propositional axioms see~\cite{troelstra_basic_2000,mancosu_introduction_2021}.

    We highlight the following
    case  $\jk : \just{s}{(A \IMP B)} \IMP (\just{t}{A} \IMP \just{\pdot{s}{t}}{B})$:
    $$
        \vlderivation
        {
            \vlin{\rr \IMP}{}
            {
                \seq{}{\just{s}{(A \IMP B)} \IMP (\just{t}{A} \IMP \just{\pdot{s}{t}}{B})}
            }
            {
                \vlin{\rr \IMP}{}
                {
                    \seq{\just{s}{(A \IMP B)}}{\just{t}{A} \IMP \just{\pdot{s}{t}}{B}}
                }
                {
                    \vliin{\rr {\boxrule[\pdot{}{}]}}{}
                    {
                        \seq{\just{s}{(A \IMP B)}, \just{t}{A}}{\just{\pdot{s}{t}}{B}}
                    }
                    {
                        \vlin{\id}{}
                        {
                            \seq{\just{s}{(A \IMP B)}, \just{t}{A}}{\just{s}{(A \IMP B)}}
                        }
                        {
                            \vlhy{}
                        }
                    }
                    {
                        \vlin{\id}{}
                        {
                            \seq{\just{s}{(A \IMP B)}, \just{t}{A}}{\just{t}{A}}
                        }
                        {
                            \vlhy{}
                        }
                    }
                }
            }
        }
    $$

    \noindent
    For the remaining cases, see \cref{ex:axscproofs} in the Appendix.
    
    The inductive case of modus ponens is simulated using $\cut$ as expected.
\end{proof}

\begin{proposition}\label{prop:soundness}
    $\seqprove{\seq{\Gamma}{A}} \implies \hprove{\form{\seq{\Gamma}{A}}}$.
\end{proposition}
\begin{proof}
    This proceeds by showing each rule is sound through axiomatic reasoning.
    The cases of the propositional rules are covered in~\cite{troelstra_basic_2000,mancosu_introduction_2021}.
    The justification rules follows from seeing a direct application of the justification axioms, except for $\rr {\boxrule[\ ]}$ for which we use \cref{prop:typeformula}.
\end{proof}

\subsection{Cut-elimination}
Our objective is to prove the following.
\begin{theorem}[Cut-elimination]\label{thm:cutelim}
    $\seqprovecut{\seq{\Gamma}{A}} \implies \seqprove{\seq{\Gamma}{A}}$.
\end{theorem}
We will show this through a cut-elimination argument.
To prove cut-elimination, we define the rank of a formula contained in a derivation.
We assume the reader is familiar with the standard notion of the degree of a formula and proof term (i.e.~the number of connectives contained in it), denoted $\deg{A}$ or $\deg{t}$ for some formula~$A$ or proof term~$t$.
This is similar to the measure used for $\iLP$~\cite{artemov_unified_2002}.

\begin{definition}[Rank of formula]
    The rank of a formula $A$ occurring in a sequent with derivation $\pi$, $\rank{\seq{\Gamma, \ann{A}}{B}}{\pi}$ or $\rank{\seq{\Gamma}{\ann{A}}}{\pi}$ is defined inductively on the derivation~$\pi$.
    There are 17 cases in total to consider, see \cref{def:rank} in the Appendix for all cases. We highlight some cases here.
    \begin{enumerate}

        \item If the preceding rule of $\pi$ is $\rr {\boxrule[\ ]}$ or $\rr {\boxrule[\pbang{}]}$:
        \[
            \vlderivation
            {
                \vlin{\rr {\boxrule[\ ]}}{}
                {
                    \seq{\Gamma}{\just{x}{A}}
                }
                {
                    \vlhtr{\pi'}{\seq{\Gamma}{A}}
                }
            }
            \qquad
            \vlderivation
            {
                \vlin{\rr {\boxrule[\pbang{}]}}{}
                {
                    \seq{\Gamma}{\just{\pbang{t}}{\just{t}{A}}}
                }
                {
                    \vlhtr{\pi'}{\seq{\Gamma}{\just{t}{A}}}
                }
            }
        \]
        then $\rank{\seq{\Gamma}{\ann{\just{x}{A}}}}{\pi} = 1 + \rank{\seq{\Gamma}{\ann{A}}}{\pi'}$ or $\rank{\seq{\Gamma}{\ann{\just{\pbang{t}}{\just{t}{A}}}}}{\pi} = \deg{\pbang{t}} + \rank{\seq{\Gamma}{\ann{\just{t}{A}}}}{\pi'}$.

        \item If the preceding rule of $\pi$ is $\rr {\boxrule[\pdot{}{}]}$:
        \[
            \vlderivation
            {
                \vliin{\rr {\boxrule[\pdot{}{}]}}{}
                {
                    \seq{\Gamma}{\just{\pdot{s}{t}}{B}}
                }
                {
                    \vlhtr{\pi_1}{\seq{\Gamma}{\just{s}{(A \IMP B)}}}
                }
                {
                    \vlhtr{\pi_2}{\seq{\Gamma}{\just{t}{A}}}
                }
            }
        \]
        then $\rank{\seq{\Gamma}{\ann{\just{\pdot{s}{t}}{B}}}}{\pi} = 1 + \rank{\seq{\Gamma}{\ann{\just{s}{(A \IMP B)}}}}{\pi_1} + \rank{\seq{\Gamma}{\ann{\just{t}{A}}}}{\pi_2}$.

        \item If the preceding rule of $\pi$ is $\rr{\boxrule[\pl{}]}$:
        \[
            \vlderivation
            {
                \vlin{\rr{\boxrule[\pl{}]}}{}
                {
                    \seq{\Gamma}{\just{\pl{t}}{(A \IMP B)}}
                }
                {
                    \vlhtr{\pi'}{\seq{\Gamma, \just{x}{A}}{\just{t}{B}}}
                }
            }
        \]
        then $\rank{\seq{\Gamma}{\ann{\just{\pl{t}}{(A \IMP B)}}}}{\pi} = 2 + \rank{\seq{\Gamma, \ann{\just{x}{A}}}{\just{t}{B}}}{\pi'} + \rank{\seq{\Gamma, \just{x}{A}}{\ann{\just{t}{B}}}}{\pi'}$.

    \end{enumerate}
\end{definition}
We are then able to define the \emph{cut-rank} of a formula and height of a deriviation.
\begin{definition}[Cut-rank and height]
    Given a cut rule instance
    \[
        \vlderivation
        {
            \vliin{\cut}{}
            {
                \seq{\Gamma, \Delta}{B}
            }
            {
                \vlhtr{\pi_1}{\seq{\Gamma}{A}}
            }
            {
                \vlhtr{\pi_2}{\seq{\Delta, A, \dots, A}{B}}
            }
        }
    \]
    the \emph{cut-rank} of the formula $A$ is given by $\max{\rank{\seq{\Gamma}{\ann{A}}}{\pi_1}, \rank{\seq{\Delta, \ann{A}, \dots, A}{B}}{\pi_2}, \dots}{\rank{\seq{\Delta, A, \dots, \ann{A}}{B}}{\pi_2}}$.

    Given a derivation $\pi$, its cut-rank, $\cutrank{\pi}$, is the maximal of the cut-ranks of formulas occurring in the derivation, if $\pi$ is cut-free, then the cut-rank is 0.

    The \emph{height} of a derivation $\pi$, $\height{\pi}$, is defined as standard, c.f.~\cite{troelstra_basic_2000,mancosu_introduction_2021}.
\end{definition}

Before proceeding with the cut-elimination, we will make use of the following technical results.
\begin{lemma}\label{lem:rankfacts}
    \begin{enumerate}
        \item For any derivation $\pi$ of a sequent $\seq{\Gamma, A}{B}$, we have $\deg{A} \leq \rank{\seq{\Gamma, \ann{A}}{B}}{\pi}$ and $\deg{B} \leq \rank{\seq{\Gamma, A}{\ann{B}}}{\pi}$;
        \item The \emph{generalised} identity rule
        \[
            \vlinf{\gid}{}
            {
                \seq{\Gamma, A}{A}
            }
            {
            }
        \]
        is \emph{admissible}, i.e.~there exists a cut-free proof $\pi$ of $\seq{\Gamma, A}{A}$ where the ranks of the formulas contained in the sequent in $\pi$ is equal to the degree of the formulas.
    \end{enumerate}
\end{lemma}
\begin{proof}
    1.~is proved by induction on the derivation~$\pi$.
    2.~is proved by induction on $\deg{A}$.
\end{proof}

The key objective for cut-elimination is to prove that cuts of a higher rank can be reduced to using cuts of smaller rank.
\begin{lemma}[Cut-reduction lemma]\label{lem:cut-reduction}
    Given a cut rule instance
    \[
        \vlderivation
        {
            \vliin{\cut}{}
            {
                \seq{\Gamma, \Delta}{B}
            }
            {
                \vlhtr{\pi_1}{\seq{\Gamma}{A}}
            }
            {
                \vlhtr{\pi_2}{\seq{\Delta, A, \dots, A}{B}}
            }
        }
    \]
    where $\cutrank{\pi_1}$ and $\cutrank{\pi_2}$ is strictly less than the cut-rank of the formula $A$, then there is a proof $\pi'$ of $\seq{\Gamma, \Delta}{B}$ where:
    \begin{enumerate}
        \item $\cutrank{\pi'}$ is strictly less than the cut-rank of the formula $A$;
        \item the rank of formulas in $\Gamma$ and $\Delta$ remains unchanged.
    \end{enumerate}
\end{lemma}
Condition~1 ensures that we are indeed reducing the cut-rank of the derivation;
condition~2 ensures that the ranks of other formulas which could depend on the formulas in $\Gamma$ are unchanged, more importantly not increased, in this derivation and other derivations for which~$\pi'$ is used to replace~$\pi$.

The proof will proceed by induction on~$\height{\pi_1} + \height{\pi_2}$  in the usual way, where we will look at the preceding rules in the derivations of $\pi_1$ and $\pi_2$.
Most of these cases are standard and can be found in~\cite{troelstra_basic_2000,mancosu_introduction_2021}.
The case that is required to be considered is the following:
\[
    \vlderivation
    {
        \vliin{\cut}{}
        {
            \seq{\Gamma, \Delta}{B}
        }
        {
            \vlhtr{\pi_1}{\seq{\Gamma}{\just{t}{A}}}
        }
        {
            \vlin{\lr {\boxrule[\ ]}}{}
            {
                \seq{\Delta, \just{t}{A}, \just{t}{A}, \dots, \just{t}{A}}{B}
            }
            {
                \vlhtr{\pi_2'}{\seq{\Delta, A, \just{t}{A}, \dots, \just{t}{A}}{B}}
            }
        }
    }
\]
To deal with the sequent $\seq{\Gamma}{\just{t}{A}}$, we utilise the following.
\begin{lemma}[Stripping Lemma]\label{lem:stripping}
    Given a derivation $\pi$ of $\seq{\Gamma}{\just{t}{A}}$, there exists a derivation $\tilde{\pi}$ of $\seq{\Gamma}{A}$ such that
    \begin{enumerate}
        \item $\rank{\seq{\Gamma}{\ann{A}}}{\tilde{\pi}} < \rank{\seq{\Gamma}{\ann{\just{t}{A}}}}{\pi}$;
        \item the ranks of formulas in $\Gamma$ are the same in both $\pi$ and $\tilde{\pi}$;
        \item $\cutrank{\tilde{\pi}} \leq \max{\cutrank{\pi}}{\rank{\seq{\Gamma}{\ann{\just{t}{A}}}}{\pi} - 1}$.
    \end{enumerate}
\end{lemma}
For proof of the Stripping Lemma, go to \cref{sec:apppt} in the Appendix.
We use the sequent obtained here to cut against the smaller formula $A$.
Condition~1 ensures that the cut-rank is smaller;
condition~2 similarly ensures that the ranks of other formulas which could depend on the formulas in $\Gamma$ are unchanged;
condition~3 ensures that the rank of new cuts introduced isn't greater than or equal to~$\rank{\seq{\Gamma}{\ann{\just{t}{A}}}}{\pi}$.

Now we can conclude with the following proof.
\begin{proof}[Proof of the Cut-reduction Lemma (\cref{lem:cut-reduction})]
    We proceed by induction on $\height{\pi_1} + \height{\pi_2}$.
    The cut case mentioned above
    is reduced in the following way.
    By the Stripping Lemma (\cref{lem:stripping}), there exists a proof $\tilde{\pi_1}$ of $\seq{\Gamma}{A}$ with the necessary conditions.
    We obtain
    \[
        \vlderivation
        {
            \vliq{\cont}{}
            {
                \seq{\Gamma, \Delta}{B}
            }
            {
                \vliin{\cut}{}
                {
                    \seq{\Gamma, \Gamma, \Delta}{B}
                }
                {
                    \vlhtr{\tilde{\pi_1}}{\seq{\Gamma}{A}}
                }
                {
                    \vliin{\cut}{\text{I.H.}}
                    {
                        \seq{\Gamma, \Delta, A}{B}
                    }
                    {
                        \vlhtr{\pi_1}{\seq{\Gamma}{\just{t}{A}}}
                    }
                    {
                        \vlhtr{\pi_2'}{\seq{\Delta, A, \just{t}{A}, \dots, \just{t}{A}}{B}}
                    }
                }
            }
        }
    \]
    where the cut on $\just{t}{A}$ is using the inductive hypothesis on height.
    The multiple lines indicate multiple (potentially zero) uses of $\cont$.
\end{proof}

Now we can conclude with the following.
\begin{proof}[Proof of the Cut-elimination Theorem (\cref{thm:cutelim})]
    Given a proof $\pi$ of $\seq{\Gamma}{A}$ in $\LJ + \cut$.
    The argument proceeds by induction on the lexicographic ordering on 
    $\cutrank{\pi}$ against the number of cuts in $\pi$ of $\cutrank{\pi}$ using the Cut-reduction Lemma (\cref{lem:cut-reduction}) as standard.
\end{proof}
\subsection{Normalisation}
In this section, we show normalisation.
For convenience, we will prove normalisation for the untyped system~$\NDJ$ but the same argument holds for the typed system~$\laJ$.
The objective is to prove the following.
\begin{theorem}[Normalisation]\label{thm:normalisation}
    $\utctx{\Gamma}{A} \implies \nutctx{\Gamma}{A}$.
\end{theorem}

We have the following translation of proofs of $\NDJ$ into $\LJ + \cut$.
\begin{lemma}\label{lem:NDtoSC}
    $\utctx{\Gamma}{A} \implies \seqprovecut{\seq{\Gamma}{A}}$.
\end{lemma}
\begin{proof}
    This is a routine proof by induction on~$\utctx{\Gamma}{A}$.
    The $\id$ rule is simulated with $\gid$,
    introduction rules are simulated with $\mathsf{r}$ rules of $\LJ$ 
    and elimination rules are simulated with a cut rule.
    We highlight the translation for the $\elimrle{\boxrule[\ ]}$ rule:
    \[
        \vlderivation
        {
            \vlin{\elimrle{\boxrule[\ ]}}{}
            {
                \utctx{\Gamma}{A}
            }
            {
                \vlhtr{}{\utctx{\Gamma}{\just{t}{A}}}
            }
        }
        \quad
        \rightsquigarrow
        \quad
        \vlderivation
        {
            \vliin{\cut}{}
            {
                \seq{\Gamma}{A}
            }
            {
                \vlhtr{}{\seq{\Gamma}{\just{t}{A}}}
            }
            {
                \vlin{\lr{\boxrule[\ ]}}{}
                {
                    \seq{\just{t}{A}}{A}
                }
                {
                    \vlin{\gid}{}
                    {
                        \seq{A}{A}
                    }
                    {
                        \vlhy{}
                    }
                }
            }
        }
        \qedhere
    \]
\end{proof}

Conversely, the following is a translation of proofs of $\LJ + \cut$ in normal derivations of~$\NDJ$.
\begin{lemma}\label{lem:SCtoND}
    $\seqprove{\seq{\Gamma}{A}} \implies \nutctx{\Gamma}{A}$.
\end{lemma}
\begin{proof}
    This is by induction on~$\seqprove{\seq{\Gamma}{A}}$.
    The propositional cases are covered in~\cite{sorensen_lectures_2006}.

    The $\mathsf{r}$~rules are simulated by introduction rules at the root of the derivation meaning a detour is not introduced.
    For the $\lr {\boxrule[\ ]}$ rule
    the inductive hypothesis gives a normal proof of $\utctx{\Gamma, B}{A}$.
    Translating this derivation into $\utctx{\Gamma, \just{t}{B}}{A}$ requires replacing the following $\id$ instances in the derivation:
    \[
        \vlinf{\id}{}
        {
            \utctx{\Gamma', B}{B}
        }
        {}
        \quad
        \rightsquigarrow
        \quad
        \vlderivation
        {
            \vlin{\elimrle{\boxrule[\ ]}}{}
            {
                \utctx{\Gamma', \just{t}{B}}{B}
            }
            {
                \vlin{\id}{}
                {
                    \utctx{\Gamma', \just{t}{B}}{\just{t}{B}}
                }
                {
                    \vlhy{}
                }
            }
        }
    \]
    and replacing the instance of the $B$ on the left being replaced by $\just{t}{B}$   
    noting that this will not introduce any detours hence also being a normal proof.
\end{proof}

\begin{proof}[Proof of Normalisation (\cref{thm:normalisation})]
    Given $\utctx{\Gamma}{A}$, apply \cref{lem:NDtoSC} to get $\seqprovecut{\seq{\Gamma}{A}}$.
    Apply the Cut-elimination Theorem (\cref{thm:cutelim}) to get $\seqprove{\seq{\Gamma}{A}}$.
    Finally apply \cref{lem:SCtoND} to get $\nutctx{\Gamma}{A}$.
\end{proof}

\section{Conclusion and Future Work}
We have introduced $\J$, a justification logic in which the proof terms of the modality are exactly the $\lambda$-terms of $\IPL$, rather than an encoding of Hilbert-style combinatory proofs as in the Logic of Proofs.
We developed both its untyped and typed natural deduction calculus along with its proof theory to achieve a normalisation result.
We have shown that $\J$ is able to self-internalise and reason about its own proof terms.
Additionally, these proof terms are able to relate to the term calculus of $\laJ$ and vice-versa. 
The results are visualised as a `tour' in \cref{fig:results}.

For future work, it is a natural question to show strong normalisation.
Methods of reducibility~\cite{girard_proofs_1993} should be appropriate. 
Extending these results to the full language of $\IPL$ including~$\BOT$ and~$\OR$ will also be appropriate.

We believe methods of embedding computational proofs into the justification logic could work for other systems such as the $\lambda\mu$-calculus~\cite{sorensen_lectures_2006}, linear logic~\cite{guerrini_proof_2004} and modal $\lambda$-calculi~\cite{kavvos_many_2016}.

Another direction is understanding $\J$ through semantics.
Justification logic can be interpreted into valuation semantics~\cite{mkrtychev_models_1997} -- in the intuitionistic setting, this can be understood through possible world semantics where results in~\cite{marti_intutionistic_2016,marin_intuitionistic_2026,KasterovicG20,kasterovic_logic_2022} should provide a good base.
On the arithmetic side, the logic of proofs ($\LP$) enjoys a connection to Peano Arithmetic ($\PA$), however the intuitionistic logic of proofs $\iLP$ does not have a complete interpretation into Heyting Arithmetic ($\HA$)~\cite{artemov_basic_2007,dashkov_arithmetical_2011}. 
Potentially, $\J$ could have a complete connection through the view of realisability.

There is a question on whether there is a meaningful connection to the underlying basic modal logic which is an extension of $\mathsf{iS4}$ with the axiom $(\BOX A \IMP \BOX B) \IMP \BOX (A \IMP B)$ which has been studied in the literature~\cite{litak_constructive_2014,van_der_giessen_uniform_2022}  through a realisation theorem, that is taking each theorem of the modal logic into a theorem of $\J$ by replacing each $\BOX$ occurrence with precisely one proof term. 

\bibliography{bibliography}

\appendix
\section{Appendix for \cref{sec:JL}}\label{sec:appjl}
\begin{proof}[Proof of \cref{thm:internalisation}]
    We proceed by induction on $\hprove A$.
    For the base cases we consider when $A$ is an axiom instance.
    
    Case $\iplax{1} : A \IMP (B \IMP A)$. 
    We have
    $$
        \hprove[\just{x}{A}, \just{y}{B}]{\just{x}{A}}
    $$
    applying the Deduction Theorem (\cref{thm:deduction}) and $\jimpi$ twice, we get
    $$
        \hprove{\just{\pl{\pl[y]{x}}}{A}}
    $$

    Case $\iplax{2} : (A \IMP (B \IMP C)) \IMP (A \IMP B) \IMP (A \IMP C)$.
    We have
    $$
        \hprove[\just{x}{(A \IMP (B \IMP C))}   ,   \just{y}{(A \IMP B)}     ,     \just{z}{A}]
        {
            \just{x}{(A \IMP (B \IMP C))}
        }
    $$
    $$
        \hprove[\just{x}{(A \IMP (B \IMP C))}   ,   \just{y}{(A \IMP B)}     ,     \just{z}{A}]
        {
            \just{y}{(A \IMP B)}
        }
    $$
    $$
        \hprove[\just{x}{(A \IMP (B \IMP C))}   ,   \just{y}{(A \IMP B)}     ,     \just{z}{A}]
        {
            \just{z}{A}
        }
    $$
    Applying $\jk$ we get
    $$
        \hprove[\just{x}{(A \IMP (B \IMP C))}   ,   \just{y}{(A \IMP B)}     ,     \just{z}{A}]
        {
            \just{\pdot{x}{z}}{(B \IMP C)}
        }
    $$
    and
    $$
        \hprove[\just{x}{(A \IMP (B \IMP C))}   ,   \just{y}{(A \IMP B)}     ,     \just{z}{A}]
        {
            \just{\pdot{y}{z}}{B}
        }
    $$
    and then apply $\jk$ on the above
    $$
        \hprove[\just{x}{(A \IMP (B \IMP C))}   ,   \just{y}{(A \IMP B)}     ,     \just{z}{A}]
        {
            \just{\pdot{(\pdot{x}{z})}{(\pdot{y}{z})}}{C}
        }
    $$
    Applying the Deduction Theorem (\cref{thm:deduction}) and $\jimpi$ several times, we get
    $$
        \hprove
        {
            \just{\pl{\pl[y]{\pl[z]{\pdot{(\pdot{x}{z})}{(\pdot{y}{z})}}}}}{((A \IMP (B \IMP C)) \IMP ((A \IMP B) \IMP (A \IMP C)))}
        }
    $$

    Case $\iplax{3} : A \IMP B \IMP (A \AND B)$.
    We have
    $$
        \hprove[ \just{x}{A}    , \just{y}{B}]
        {
            \just{x}{A}
        }
    $$
    and
    $$
        \hprove[ \just{x}{A}    , \just{y}{B}]
        {
            \just{y}{B}
        }
    $$
    Propositional reasoning with $\jandi$ gives
    $$
        \hprove[ \just{x}{A}    , \just{y}{B}]
        {
            \just{\pai{x}{y}}{(A \AND B)}
        }
    $$
    Applying the Deduction Theorem (\cref{thm:deduction}) and $\jimpi$ several times, we get
    $$
        \hprove
        {
            \just{\pl\pl[y]\pai{x}{y}}{(A \IMP B \IMP (A \AND B))}
        }
    $$

    Case $\iplax{4} : (A \AND B) \IMP A$.
    Using the $\jandel$ axiom we have
    $$
        \hprove
        {
            \just{x}{(A \AND B)}
            \IMP
            \just{\pael{x}}{A}
        }
    $$
    By the $\jimpi$ axiom, we get
    $$
        \hprove
        {
            \just{\pl{\pael{x}}}
            {
            ((A \AND B)
            \IMP
            A)
            }
        }
    $$
    The Case $\iplax{5} : (A \AND B) \IMP B$ follows from a similar argument.

    Case $A$ is an instance of the $\jk ,\jimpi, \jandi , \jandel , \jander ,\jfax$, we have
    $$
        A \colonequals \just{t_1}{A_1} \IMP \just{t_2}{A_2} \IMP \dots \IMP \just{t_n}{A_n} 
    $$
    for some proof terms $t_1, \dots, t_n$ and formulas $A_1, \dots, A_n$.
    For each $i = 1, \dots, n$ we have
    $$
        \hprove[\just{x_1}{\just{t_1}{A_1}}, \just{x_2}{\just{t_2}{A_2}}, \dots \just{x_{n-1}}{\just{t_{n-1}}{A_{n-1}}}]
        {\just{x_i}{\just{t_i}{A_i}}}
    $$
    Using the $\jtax$ axiom, we get
    $$
        \hprove[\just{x_1}{\just{t_1}{A_1}}, \just{x_2}{\just{t_2}{A_2}}, \dots \just{x_{n-1}}{\just{t_{n-1}}{A_{n-1}}}]
        {\just{t_i}{A_i}}
    $$
    Using the axiom $A$ and $\mp$ repeatedly, we get
    $$
        \hprove[\just{x_1}{\just{t_1}{A_1}}, \just{x_2}{\just{t_2}{A_2}}, \dots \just{x_{n-1}}{\just{t_{n-1}}{A_{n-1}}}]
        {
            \just{t_n}{A_n}
        }
    $$
    Using the $\jfax$ axiom, we have
    $$
        \hprove[\just{x_1}{\just{t_1}{A_1}}, \just{x_2}{\just{t_2}{A_2}}, \dots \just{x_{n-1}}{\just{t_{n-1}}{A_{n-1}}}]
        {
            \just{\pbang{t_n}}{\just{t_n}{A_n}}
        }
    $$
    Applying the Deduction Theorem (\cref{thm:deduction}) and $\jimpi$ several times, we get
    $$
        \hprove
        {
            \just{\pl[x_1]{\pl[x_2]{ \dots \pl[x_{n-1}}}]{\pbang{t_n}}}{(\just{t_1}{A_1} \IMP \just{t_2}{A_2} \IMP \dots \IMP \just{t_n}{A_n} )}
        }
    $$

    Case $\jtax : \just{t}{A} \IMP A$.
    Using the $\jtax$ axiom we get
    $$
        \hprove
        {
            \just{x}{\just{t}{A}}
            \IMP
            \just{t}{A}
        }
    $$
    By the $\jimpi$ axiom, we get
    $$
        \hprove
        {
            \just{\pl{t}}{(\just{t}{A} \IMP A)}
        }
    $$

    For the inductive case, a conclusion of $\mp$
    $$
        \vliinf{\mp}{}{\hprove{B}}{\hprove{A}}{\hprove{A \IMP B}}
    $$
    By the inductive hypothesis, we have $\hprove{\just{t}{A}}$ and $\hprove{\just{s}{(A \IMP B)}}$ for some closed proof terms $s$ and $t$.
    Using the $\jk$ axiom, we have
    $$
        \hprove{\just{\pdot{s}{t}}{B}}
    $$
    noting that $\pdot{s}{t}$ is a closed proof term.
\end{proof}

\section{Appendix for \cref{sec:ND}}\label{sec:appnd}
\begin{example}\label{ex:ndjustax}
    $$
        \vlderivation
        {
            \vlin{\intrle{\IMP}}{}
            {
                \utctx{}{(\just{x}{A} \IMP \just{t}{B}) \IMP \just{\pl{t}}{(A \IMP B)}}
            }
            {
                \vlin{\intrle{\boxrule[\pl{}]}}{}
                {
                    \utctx{\just{x}{A} \IMP \just{t}{B}}{\just{\pl{t}}{(A \IMP B)}}
                }
                {
                    \vliin{\elimrle{\IMP}}{}
                    {
                        \utctx{\just{x}{A} \IMP \just{t}{B}, \just{x}{A}}{\just{t}{B}}
                    }
                    {
                        \vlin{\id}{}
                        {
                            \utctx{\just{x}{A} \IMP \just{t}{B}, \just{x}{A}}{\just{x}{A} \IMP \just{t}{B}}
                        }
                        {
                            \vlhy{}
                        }
                    }
                    {
                        \vlin{\id}{}
                        {
                            \utctx{\just{x}{A} \IMP \just{t}{B}, \just{x}{A}}{\just{x}{A}}
                        }
                        {
                            \vlhy{}
                        }
                    }
                }
            }
        }
    $$
    $$
        \vlderivation
        {
            \vlin{\intrle{\IMP}}{}
            {
                \utctx{}{\just{s}{A} \IMP \just{t}{B} \IMP \just{\pai{s}{t}}{(A \AND B)}}
            }
            {
                \vlin{\intrle{\IMP}}{}
                {
                    \utctx{\just{s}{A}}{\just{t}{B} \IMP \just{\pai{s}{t}}{(A \AND B)}}
                }
                {
                    \vliin{\intrle{\boxrule[\pai{}{}]}}{}
                    {
                        \utctx{\just{s}{A}, \just{t}{B}}{\just{\pai{s}{t}}{(A \AND B)}}
                    }
                    {
                        \vlin{\id}{}
                        {
                            \utctx{\just{s}{A}, \just{t}{B}}{\just{s}{A}}
                        }
                        {
                            \vlhy{}
                        }
                    }
                    {
                        \vlin{\id}{}
                        {
                            \utctx{\just{s}{A}, \just{t}{B}}{\just{t}{B}}
                        }
                        {
                            \vlhy{}
                        }
                    }
                }
            }
        }
    $$

    $$
        \vlderivation
        {
            \vlin{\intrle{\IMP}}{}
            {
                \utctx{}{\just{t}{(A \AND B)} \IMP \just{\pael{t}}{A}}
            }
            {
                \vlin{\intrle{\boxrule[\pael{}]}}{}
                {
                    \utctx{\just{t}{(A \AND B)}}{\just{\pael{t}}{A}}
                }
                {
                    \vlin{\id}{}
                    {
                        \utctx{\just{t}{(A \AND B)}}{\just{t}{(A \AND B)}}
                    }
                    {
                        \vlhy{}
                    }
                }
            }
        }
    $$

    $$
        \vlderivation
        {
            \vlin{\intrle{\IMP}}{}
            {
                \utctx{}{\just{t}{(A \AND B)} \IMP \just{\paer{t}}{B}}
            }
            {
                \vlin{\intrle{\boxrule[\paer{}]}}{}
                {
                    \utctx{\just{t}{(A \AND B)}}{\just{\paer{t}}{B}}
                }
                {
                    \vlin{\id}{}
                    {
                        \utctx{\just{t}{(A \AND B)}}{\just{t}{(A \AND B)}}
                    }
                    {
                        \vlhy{}
                    }
                }
            }
        }
    $$

    $$
        \vlderivation
        {
            \vlin{\intrle{\IMP}}{}
            {
                \utctx{}{\just{t}{A} \IMP \just{\pbang{t}}{\just{t}{A}}}
            }
            {
                \vlin{\intrle{\boxrule[\pbang{}]}}{}
                {
                    \utctx{\just{t}{A}}{\just{\pbang{t}}{\just{t}{A}}}
                }
                {
                    \vlin{\id}{}
                    {
                        \utctx{\just{t}{A}}{\just{t}{A}}
                    }
                    {
                        \vlhy{}
                    }
                }
            }
        }
    $$

    $$
        \vlderivation
        {
            \vlin{\intrle{\IMP}}{}
            {
                \utctx{}{\just{t}{A} \IMP A}
            }
            {
                \vlin{\elimrle{\boxrule[\ ]}}{}
                {
                    \utctx{\just{t}{A}}{A}
                }
                {
                    \vlin{\id}{}
                    {
                        \utctx{\just{t}{A}}{\just{t}{A}}
                    }
                    {
                        \vlhy{}
                    }
                }
            }
        }
    $$
\end{example}

\begin{example}\label{ex:typjustax}
    $$
        \vlderivation
        {
            \vlin{\intrle{\IMP}}{}
            {
                \ctx{}{\fm{\la[\var[1]]{\la[\var[2]]{\boxapp{\var[1]}{\var[2]}}}}{\just{s}{(A \IMP B)} \IMP \just{t}{A} \IMP \just{\pdot{s}{t}}{B}}}
            }
            {
                \vlin{\intrle{\IMP}}{}
                {
                    \ctx{\fm{\var[1]}{\just{s}{(A \IMP B)}}}{\fm{\la[\var[2]]{\boxapp{\var[1]}{\var[2]}}}{\just{t}{A} \IMP \just{\pdot{s}{t}}{B}}}
                }
                {
                    \vliin{\intrle{\boxrule[\pdot{}{}]}}{}
                    {
                        \ctx{\fm{\var[1]}{\just{s}{(A \IMP B)}}, \fm{\var[2]}{\just{t}{A}}}{ \fm{\boxapp{\var[1]}{\var[2]}}{\just{\pdot{s}{t}}{B}}}
                    }
                    {
                        \vlin{\id}{}
                        {
                            \ctx{\fm{\var[1]}{\just{s}{(A \IMP B)}}, \fm{\var[2]}{\just{t}{A}}}{ \fm{\var[1]}{\just{s}{(A \IMP B)}}}
                        }
                        {
                            \vlhy{}
                        }
                    }
                    {
                        \vlin{\id}{}
                        {
                            \ctx{\fm{\var[1]}{\just{s}{(A \IMP B)}}, \fm{\var[2]}{\just{t}{A}}}{ \fm{\var[2]}{\just{t}{B}}}
                        }
                        {
                            \vlhy{}
                        }
                    }
                }
            }
        }
    $$

    $$
        \vlderivation
        {
            \vlin{\intrle{\IMP}}{}
            {
                \ctx{}{\fm{\la[\var[1]]{\la[\var[2]]{\boxandc{\var[1]}{\var[2]}}}}{\just{s}{A} \IMP \just{t}{B} \IMP \just{\pai{s}{t}}{(A \AND B)}}}
            }
            {
                \vlin{\intrle{\IMP}}{}
                {
                    \ctx{\fm{\var[1]}{\just{s}{A}}}{\fm{\la[\var[2]]{\boxandc{\var[1]}{\var[2]}}}{\just{t}{B} \IMP \just{\pai{s}{t}}{(A \AND B)}}}
                }
                {
                    \vliin{\intrle{\boxrule[\pai{}{}]}}{}
                    {
                        \ctx{\fm{\var[1]}{\just{s}{A}}, \fm{\var[2]}{\just{t}{B}}}{\fm{\boxandc{\var[1]}{\var[2]}}{\just{\pai{s}{t}}{(A \AND B)}}}
                    }
                    {
                        \vlin{\id}{}
                        {
                            \ctx{\fm{\var[1]}{\just{s}{A}}, \fm{\var[2]}{\just{t}{B}}}{\fm{\var[1]}{\just{s}{A}}}
                        }
                        {
                            \vlhy{}
                        }
                    }
                    {
                        \vlin{\id}{}
                        {
                            \ctx{\fm{\var[1]}{\just{s}{A}}, \fm{\var[2]}{\just{t}{B}}}{\fm{\var[2]}{\just{t}{B}}}
                        }
                        {
                            \vlhy{}
                        }
                    }
                }
            }
        }
    $$

    $$
        \vlderivation
        {
            \vlin{\intrle{\IMP}}{}
            {
                \ctx{}{\fm{\la{\boxprojl{\var}}}{\just{t}{(A \AND B)} \IMP \just{\pael{t}}{A}}}
            }
            {
                \vlin{\intrle{\boxrule[\pael{}]}}{}
                {
                    \ctx{\fm{\var}{\just{t}{(A \AND B)}}}{\fm{\boxprojl{\var}}{\just{\pael{t}}{A}}}
                }
                {
                    \vlin{\id}{}
                    {
                        \ctx{\fm{\var}{\just{t}{(A \AND B)}}}{\fm{\var}{\just{t}{(A \AND B)}}}
                    }
                    {
                        \vlhy{}
                    }
                }
            }
        }
    $$

    $$
        \vlderivation
        {
            \vlin{\intrle{\IMP}}{}
            {
                \ctx{}{\fm{\la{\boxprojr{\var}}}{\just{t}{(A \AND B)} \IMP \just{\paer{t}}{B}}}
            }
            {
                \vlin{\intrle{\boxrule[\paer{}]}}{}
                {
                    \ctx{\fm{\var}{\just{t}{(A \AND B)}}}{\fm{\boxprojr{\var}}{\just{\paer{t}}{B}}}
                }
                {
                    \vlin{\id}{}
                    {
                        \ctx{\fm{\var}{\just{t}{(A \AND B)}}}{\fm{\var}{\just{t}{(A \AND B)}}}
                    }
                    {
                        \vlhy{}
                    }
                }
            }
        }
    $$

    $$
        \vlderivation
        {
            \vlin{\intrle{\IMP}}{}
            {
                \ctx{}{\fm{\la{\bang{\var}}}{\just{t}{A} \IMP \just{\pbang{t}}{\just{t}{A}}}}
            }
            {
                \vlin{\intrle{\boxrule[\pbang{}]}}{}
                {
                    \ctx{\fm{\var}{\just{t}{A}}}{\fm{\bang{\var}}{\just{\pbang{t}}{\just{t}{A}}}}
                }
                {
                    \vlin{\id}{}
                    {
                        \ctx{\fm{\var}{\just{t}{A}}}{\fm{\var}{\just{t}{A}}}
                    }
                    {
                        \vlhy{}
                    }
                }
            }
        }
    $$

    $$
        \vlderivation
        {
            \vlin{\intrle{\IMP}}{}
            {
                \ctx{}{\fm{\la{\unbox{\var}}}{\just{t}{A} \IMP A}}
            }
            {
                \vlin{\elimrle{\boxrule[\ ]}}{}
                {
                    \ctx{\fm{\var}{\just{t}{A}}}{\fm{\unbox{\var}}{A}}
                }
                {
                    \vlin{\id}{}
                    {
                        \ctx{\fm{\var}{\just{t}{A}}}{\fm{\var}{\just{t}{A}}}
                    }
                    {
                        \vlhy{}
                    }
                }
            }
        }
    $$
\end{example}

\begin{example}\label{ex:inttyped}
    The following are theorems obtained from $\cref{thm:internalisation}$:
    \begin{enumerate}
        \item $\hprove{\just{\pl[\pv[\just{s}{(A \IMP B)}]]{\pl[\pv[\just{t}{A}]] {\pbang{(\pdot{s}{t})}}}}{(\just{s}{(A \IMP B)} \IMP \just{t}{A} \IMP \just{\pdot{s}{t}}{B})}}$;

        \item $\hprove{\just{\pl[\pv[\just{t}{A}]]{\pbang{\pbang{t}}}}{(\just{t}{A} \IMP \just{\pbang{t}}{\just{t}{A}})}}$;

        \item $\hprove{\just{\pl[\pv[\just{t}{A}]]{t}}{(\just{t}{A} \IMP A)}}$;

        \item $\hprove{\just{\pl[\pv[\just{x}{A} \IMP \just{t}{B}]]{\pbang{\pl{t}}}}{((\just{x}{A} \IMP \just{t}{B}) \IMP \just{\pl{t}}{(A \IMP B)})}}$;

        \item $\hprove{\just{\pl[\pv[\just{t}{(A \AND B)}]]{\pbang{\pael{t}}}}{(\just{t}{(A \AND B)} \IMP \just{\pael{t}}{A})}}$;

        \item $\hprove{\just{\pl[\pv[\just{t}{(A \AND B)}]]{\pbang{\paer{t}}}}{(\just{t}{(A \AND B)} \IMP \just{\paer{t}}{B})}}$;

        \item $\hprove{\just{\pl[\pv[\just{s}{A}]]{\pl[\pv[\just{t}{B}]]{\pbang{\pai{s}{t}}}}}{(\just{s}{A} \IMP \just{t}{B} \IMP \just{\pai{s}{t}}{(A \AND B)})}}$.
    \end{enumerate}
\end{example}

\begin{proof}[Proof of \cref{thm:tmtoptm}]
    This proceeds by induction on~$\ctx{\fm{\var[A_1]}{A_1}, \dots, \fm{\var[A_n]}{A_n}}{\fm{M(\var[A_1], \dots, \var[A_n])}{A}}$
    unfolding \cref{def:ttpt} at each rule. 

    For the base case~$\ctx{\fm{\var[A_1]}{A_1}, \dots, \fm{\var[A_n]}{A_n}}{\fm{\var[A_i]}{A_i}}$,
    we have $\hprove[\just{\ttpt{\var[A_1]}}{A_1}, \dots, \just{\ttpt{\var[A_n]}}{A_n}]{\just{\ttpt{\var[A_i]}}{A_i}}$ trivially.

    For the inductive cases we consider the previous rule applied.
    When the previous rules are the propositional ones, it is a simple unwinding of \cref{def:ttpt} and using the justification axioms to show soundness.

    When the preceding rule is $\intrle{\boxrule[\pbang{}]}$:
    $$
        \vlderivation
        {
            \vlin{\intrle{\boxrule[\pbang{}]}}{}
            {
                \ctx{\fm{\var[A_1]}{A_1}, \dots, \fm{\var[A_n]}{A_n}}{\fm{\bang{M(\var[A_1], \dots, \var[A_n])}}{\just{\pbang{t}}{\just{t}{A}}}}
            }
            {
                \vlhtr{}{\ctx{\fm{\var[A_1]}{A_1}, \dots, \fm{\var[A_n]}{A_n}}{\fm{M(\var[A_1], \dots, \var[A_n])}{\just{t}{A}}}}
            }
        }
    $$
    by the inductive hypothesis,
    $\hprove[\just{\ttpt{\var[A_1]}}{A_1}, \dots, \just{\ttpt{\var[A_n]}}{A_n}]{\just{\ttpt{M(\var[A_1], \dots, \var[A_n])}}{\just{t}{A}}}$, noting that $\boxterm{\ttpt{M(\var[A_1], \dots, \var[A_n])}} = t$.
    Applying the $\jtax$ and $\jfax$ axioms, we have
    $\hprove[\just{\ttpt{\var[A_1]}}{A_1}, \dots, \just{\ttpt{\var[A_n]}}{A_n}]{\just{\pbang{t}}{\just{t}{A}}}$, noting that $\ttpt{\bang{M(\var[A_1], \dots, \var[A_n])}} = \pbang{t}$.
    The other cases are similar.
\end{proof}

\section{Appendix for \cref{sec:pt}}\label{sec:apppt}
\begin{example}\label{ex:axscproofs}
    $\jimpi : (\just{x}{A} \IMP \just{t}{B}) \IMP \just{\pl{t}}{(A \IMP B)}$:
    $$
        \vlderivation
        {
            \vlin{\rr \IMP}{}
            {
                \seq{}{(\just{x}{A} \IMP \just{t}{B}) \IMP \just{\pl{t}}{(A \IMP B)}}
            }
            {
                \vlin{\rr{\boxrule[\pl{}]}}{}
                {
                    \seq{\just{x}{A} \IMP \just{t}{B}}{\just{\pl{t}}{(A \IMP B)}}
                }
                {
                    \vliin{\lr \IMP}{}
                    {
                        \seq{\just{x}{A} \IMP \just{t}{B}, \just{x}{A}}{\just{t}{B}}
                    }
                    {
                        \vlin{\id}{}
                        {
                            \seq{\just{t}{B}, \just{x}{A}}{\just{t}{B}}
                        }
                        {
                            \vlhy{}
                        }
                    }
                    {
                        \vlin{\id}{}
                        {
                            \seq{\just{t}{B}, \just{x}{A}}{\just{x}{A}}
                        }
                        {
                            \vlhy{}
                        }
                    }
                }
            }
        }
    $$

    $\jandi : \just{s}{A} \IMP (\just{t}{B} \IMP \just{\pai{s}{t}}{(A \AND B)})$:
    $$
        \vlderivation
        {
            \vlin{\rr \IMP}{}
            {
                \seq{}{\just{s}{A} \IMP (\just{t}{B} \IMP \just{\pai{s}{t}}{(A \AND B)})}
            }
            {
                \vlin{\rr \IMP}{}
                {
                    \seq{\just{s}{A}}{\just{t}{B} \IMP \just{\pai{s}{t}}{(A \AND B)}}
                }
                {
                    \vliin{\rr{\boxrule[\pai{}{}]}}{}
                    {
                        \seq{\just{s}{A}, \just{t}{B}}{ \just{\pai{s}{t}}{(A \AND B)}}
                    }
                    {
                        \vlin{\id}{}
                        {
                            \seq{\just{s}{A}, \just{t}{B}}{ \just{s}{A}}
                        }
                        {
                            \vlhy{}
                        }
                    }
                    {
                        \vlin{\id}{}
                        {
                            \seq{\just{s}{A}, \just{t}{B}}{ \just{t}{B}}
                        }
                        {
                            \vlhy{}
                        }
                    }
                }
            }
        }
    $$

    $\jandel : \just{t}{(A \AND B)} \IMP \just{\pael{t}}{A}$:
    $$
        \vlderivation
        {
            \vlin{\rr \IMP}{}
            {
                \seq{}{\just{t}{(A \AND B)} \IMP \just{\pael{t}}{A}}
            }
            {
                \vlin{}{\rr{\boxrule[\pael{}]}}
                {
                    \seq{\just{t}{(A \AND B)}}{\just{\pael{t}}{A}}
                }
                {
                    \vlin{\id}{}
                    {
                        \seq{\just{t}{(A \AND B)}}{\just{t}{(A \AND B)}}
                    }
                    {
                        \vlhy{}
                    }
                }
            }
        }
    $$

    $\jander : \just{t}{(A \AND B)} \IMP \just{\paer{t}}{B}$ is similar.

    $\jtax : \just{t}{A} \IMP A$:
    $$
        \vlderivation
        {
            \vlin{\lr \IMP}{}
            {
                \seq{}{\just{t}{A} \IMP A}
            }
            {
                \vlin{\lr {\boxrule[\ ]}}{}
                {
                    \seq{\just{t}{A}}{ A}
                }
                {
                    \vlin{\id}{}
                    {
                        \seq{A}{A}
                    }
                    {
                        \vlhy{}
                    }
                }
            }
        }
    $$

    $\jfax : \just{t}{A} \IMP \just{\pbang{t}}{\just{t}{A}}$:
    $$
        \vlderivation
        {
            \vlin{\lr \IMP}{}
            {
                \seq{}{\just{t}{A} \IMP \just{\pbang{t}}{\just{t}{A}}}
            }
            {
                \vlin{\rr {\boxrule[\pbang{}]}}{}
                {
                    \seq{\just{t}{A}}{ \just{\pbang{t}}{\just{t}{A}}}
                }
                {
                    \vlin{\id}{}
                    {
                        \seq{\just{t}{A}}{\just{t}{A}}
                    }
                    {
                        \vlhy{}
                    }
                }
            }
        }
    $$
\end{example}

\begin{definition}[Rank of formula]\label{def:rank}
    The rank of a formula $A$ occuring in a sequent with derivation $\pi$, $\rank{\seq{\Gamma, \ann{A}}{B}}{\pi}$ or $\rank{\seq{\Gamma}{\ann{A}}}{\pi}$ is defined inductively on the derivation~$\pi$:

    \begin{enumerate}
        \item If $\pi$ is a conclusion of $\id$, then $\rank{\seq{\Gamma, \ann{A}}{B}}{\pi} = \deg{A}$ or $\rank{\seq{\Gamma}{\ann{A}}}{\pi} = \deg{A}$.
        
        \item If the preceding rule of $\pi$ is $\cont$:
        $$
            \vlderivation
            {
                \vlin{\cont}{}
                    {
                        \seq{\Gamma, A}{B}
                    }
                    {
                        \vlhtr{\pi'}{\seq{\Gamma, A, A}{B}}
                    }
            }
        $$
        then $\rank{\seq{\Gamma, \ann{A}}{B}}{\pi} = \max{\rank{\seq{\Gamma, \ann{A}, A}{B}}{\pi'}}{\rank{\seq{\Gamma, A, \ann{A}}{B}}{\pi'}}$.

        \item If the preceding rule of $\pi$ is $\lr \AND$:
        $$
            \vlderivation
            {
                \vlin{\lr \AND}{}
                {
                    \seq{\Gamma, A \AND B}{C}
                }
                {
                    \vlhtr{\pi'}{\seq{\Gamma, A, B}{C}}
                }
            }
        $$
        then $\rank{\seq{\Gamma, \ann{A \AND B}}{C}}{\pi} = 1 + \rank{\seq{\Gamma, \ann{A}, B}{C}}{\pi'} + \rank{\seq{\Gamma, A, \ann{B}}{C}}{\pi'}$.

        \item If the preceding rule of $\pi$ is $\rr \AND$:
        $$
            \vlderivation
            {
                \vliin{\rr \AND}{}
                {
                    \seq{\Gamma}{A \AND B}
                }
                {
                    \vlhtr{\pi_1}{\seq{\Gamma}{A}}
                }
                {
                    \vlhtr{\pi_2}{\seq{\Gamma}{B}}
                }
            }
        $$
        then $\rank{\seq{\Gamma}{\ann{A \AND B}}}{\pi} = 1 + \rank{\seq{\Gamma}{\ann{A}}}{\pi_1} + \rank{\seq{\Gamma}{\ann{B}}}{\pi_2}$.

        \item If the preceding rule of $\pi$ is $\lr \IMP$:
        $$
            \vlderivation
            {
                \vliin{\lr \IMP}{}
                {
                    \seq{\Gamma, A \IMP B}{C}
                }
                {
                    \vlhtr{\pi_1}{\seq{\Gamma, B}{C}}
                }
                {
                    \vlhtr{\pi_2}{\seq{\Gamma}{A}}
                }
            }
        $$
        then $\rank{\seq{\Gamma, \ann{A \IMP B}}{C}}{\pi} = 1 + \rank{\seq{\Gamma, \ann{B}}{C}}{\pi_1} + \rank{\seq{\Gamma}{\ann{A}}}{\pi_2}$ and $\rank{\seq{\Gamma, A \IMP B}{\ann{C}}}{\pi} = \rank{\seq{\Gamma, B}{\ann{C}}}{\pi_1}$.

        \item If the preceding rule of $\pi$ is $\rr \IMP$:
        $$
            \vlderivation
            {
                \vlin{\rr \IMP}{}
                {
                    \seq{\Gamma}{A \IMP B}
                }
                {
                    \vlhtr{\pi'}{\seq{\Gamma, A}{B}}
                }
            }
        $$
        then $\rank{\seq{\Gamma}{\ann{A \IMP B}}}{\pi} = 1 + \rank{\seq{\Gamma, \ann{A}}{B}}{\pi'} + \rank{\seq{\Gamma, A}{\ann{B}}}{\pi'}$.

        \item If the preceding rule of $\pi$ is $\lr {\boxrule[\ ]}$:
        $$
            \vlderivation
            {
                \vlin{\lr {\boxrule[\ ]}}{}
                {
                    \seq{\Gamma, \just{t}{A}}{B}
                }
                {
                    \vlhtr{\pi'}{\seq{\Gamma, A}{B}}
                }
            }
        $$
        then $\rank{\seq{\Gamma, \ann{\just{t}{A}}}{B}}{\pi} = \deg{t} + \rank{\seq{\Gamma, \ann{A}}{B}}{\pi'}$.

        \item If the preceding rule of $\pi$ is $\rr {\boxrule[\ ]}$:
        $$
            \vlderivation
            {
                \vlin{\rr {\boxrule[\ ]}}{}
                {
                    \seq{\Gamma}{\just{x}{A}}
                }
                {
                    \vlhtr{\pi'}{\seq{\Gamma}{A}}
                }
            }
        $$
        then $\rank{\seq{\Gamma}{\ann{\just{x}{A}}}}{\pi} = 1 + \rank{\seq{\Gamma}{\ann{A}}}{\pi'}$.

        \item If the preceding rule of $\pi$ is $\rr {\boxrule[\pbang{}]}$:
        $$
            \vlderivation
            {
                \vlin{\rr {\boxrule[\pbang{}]}}{}
                {
                    \seq{\Gamma}{\just{\pbang{t}}{\just{t}{A}}}
                }
                {
                    \vlhtr{\pi'}{\seq{\Gamma}{\just{t}{A}}}
                }
            }
        $$
        then $\rank{\seq{\Gamma}{\ann{\just{\pbang{t}}{\just{t}{A}}}}}{\pi} = \deg{\pbang{t}} + \rank{\seq{\Gamma}{\ann{\just{t}{A}}}}{\pi'}$.

        \item If the preceding rule of $\pi$ is $\rr {\boxrule[\pdot{}{}]}$:
        $$
            \vlderivation
            {
                \vliin{\rr {\boxrule[\pdot{}{}]}}{}
                {
                    \seq{\Gamma}{\just{\pdot{s}{t}}{B}}
                }
                {
                    \vlhtr{\pi_1}{\seq{\Gamma}{\just{s}{(A \IMP B)}}}
                }
                {
                    \vlhtr{\pi_2}{\seq{\Gamma}{\just{t}{A}}}
                }
            }
        $$
        then $\rank{\seq{\Gamma}{\ann{\just{\pdot{s}{t}}{B}}}}{\pi} = 1 + \rank{\seq{\Gamma}{\ann{\just{s}{(A \IMP B)}}}}{\pi_1} + \rank{\seq{\Gamma}{\ann{\just{t}{A}}}}{\pi_2}$.

        \item If the preceding rule of $\pi$ is $\rr{\boxrule[\pl{}]}$:
        $$
            \vlderivation
            {
                \vlin{\rr{\boxrule[\pl{}]}}{}
                {
                    \seq{\Gamma}{\just{\pl{t}}{(A \IMP B)}}
                }
                {
                    \vlhtr{\pi'}{\seq{\Gamma, \just{x}{A}}{\just{t}{B}}}
                }
            }
        $$
        then $\rank{\seq{\Gamma}{\ann{\just{\pl{t}}{(A \IMP B)}}}}{\pi} = 2 + \rank{\seq{\Gamma, \ann{\just{x}{A}}}{\just{t}{B}}}{\pi'} + \rank{\seq{\Gamma, \just{x}{A}}{\ann{\just{t}{B}}}}{\pi'}$.

        \item If the preceding rule of $\pi$ is $\rr{\boxrule[\pael{}]}$:
        $$
            \vlderivation
            {
                \vlin{\rr{\boxrule[\pael{}]}}{}
                {
                    \seq{\Gamma}{\just{\pael{t}}{A}}
                }
                {
                    \vlhtr{\pi'}{\seq{\Gamma}{\just{t}{(A \AND B)}}}
                }
            }
        $$
        then $\rank{\seq{\Gamma}{\ann{\just{\pael{t}}{A}}}}{\pi} = 1 + \rank{\seq{\Gamma}{\ann{\just{t}{(A \AND B)}}}}{\pi'}$.

        \item If the preceding rule of $\pi$ is $\rr{\boxrule[\paer{}]}$:
        $$
            \vlderivation
            {
                \vlin{\rr{\boxrule[\paer{}]}}{}
                {
                    \seq{\Gamma}{\just{\paer{t}}{B}}
                }
                {
                    \vlhtr{\pi'}{\seq{\Gamma}{\just{t}{(A \AND B)}}}
                }
            }
        $$
        then $\rank{\seq{\Gamma}{\ann{\just{\paer{t}}{B}}}}{\pi} = 1 + \rank{\seq{\Gamma}{\ann{\just{t}{(A \AND B)}}}}{\pi'}$.

        \item  If the preceding rule of $\pi$ is $\rr{\boxrule[\pai{}{}]}$:
        $$
            \vlderivation
            {
                \vliin{\rr{\boxrule[\pai{}{}]}}{}
                {
                    \seq{\Gamma}{\just{\pai{s}{t}}{(A \AND B)}}
                }
                {
                    \vlhtr{\pi_1}{\seq{\Gamma}{\just{s}{A}}}
                }
                {
                    \vlhtr{\pi_2}{\seq{\Gamma}{\just{t}{B}}}
                }
            }
        $$
        then $\rank{\seq{\Gamma}{\ann{\just{\pai{s}{t}}{(A \AND B)}}}}{\pi} = 2 + \rank{\seq{\Gamma}{\ann{\just{s}{A}}}}{\pi_1} + \rank{\seq{\Gamma}{\ann{\just{t}{B}}}}{\pi_2}$.

        \item  If the preceding rule of $\pi$ is $\cut$:
        $$
            \vlderivation
            {
                \vliin{\cut}{}
                {
                    \seq{\Gamma, A_1, \Delta, A_2}{B}
                }
                {
                    \vlhtr{\pi_1}{\seq{\Gamma, A_1}{A}}
                }
                {
                    \vlhtr{\pi_2}{\seq{\Delta, A_2, A, \dots, A}{B}}
                }
            }
        $$
        then $\rank{\seq{\Gamma, \ann{A_1}, \Delta, A_2}{B}}{\pi} = \rank{\seq{\Gamma, \ann{A_1}}{A}}{\pi_1}$ and $\rank{\seq{\Gamma, A_1, \Delta, \ann{A_2}}{B}}{\pi} = \rank{\seq{\Delta, \ann{A_2}, A, \dots, A}{B}}{\pi_2}$.

        \item  If the preceding rule of $\pi$ is of the form:
        $$
            \vlderivation
            {
                \vlin{\rle}{}
                {
                    \seq{\Gamma}{A}
                }
                {
                    \vlhtr{\pi'}{\seq{\Gamma'}{A}}
                }
            }
        $$
        then $\rank{\seq{\Gamma}{\ann{A}}}{\pi} = \rank{\seq{\Gamma'}{\ann{A}}}{\pi'}$.

        \item  If the preceding rule of $\pi$ is of the form:
        $$
            \vlderivation
            {
                \vliiin{\rle}{}
                {
                    \seq{\Gamma, A}{B}
                }
                {
                    \vlhtr{\pi_1}{\seq{\Gamma, A}{B_1}}
                }
                {
                    \vlhy{\cdots}
                }
                {
                    \vlhtr{\pi_n}{\seq{\Gamma, A}{B_n}}
                }
            }
        $$
        for $n = 1, 2$, then $\rank{\seq{\Gamma, \ann{A}}{B}}{\pi} = \max{\rank{\seq{\Gamma, \ann{A}}{B_1}}{\pi_1}, \dots}{\rank{\seq{\Gamma, \ann{A}}{B_n}}{\pi_n}}$.
    \end{enumerate}
\end{definition}

\begin{proof}[Proof of the Stripping Lemma (\cref{lem:stripping})]
    We proceed by induction on $\pi$.

    For the base case $\pi$ is of the form
    $$
        \vlinf{\id}{}
        {
            \seq{\Gamma, \just{t}{A}}{\just{t}{A}}
        }
        {}
    $$
    we construct $\tilde\pi$ through the following
    $$
        \vlderivation
        {
            \vlin{\lr {\boxrule[\ ]}}{}
            {
                \seq{\Gamma, \just{t}{A}}{A}
            }
            {
                \vlin{\gid}{}
                {
                    \seq{\Gamma, A}{A}
                }
                {
                    \vlhy{}
                }
            }
        }
    $$
    noting that the conditions 1-3 are met.

    For the inductive cases, we look at the preceding rule of $\pi$.
    The cases of the rules $\cont, \lr \IMP, \lr \AND, \lr {\boxrule[\ ]}, \cut$ are straightforward as these rules do not act on the formula $\just{t}{A}$ in the succedent of the sequent.

    If the preceding rule is $\rr {\boxrule[\ ]}$ or $\rr {\boxrule[\pbang{}]}$:
    $$
        \vlderivation
        {
            \vlin{\rr {\boxrule[\ ]}}{}
            {
                \seq{\Gamma}{\just{x}{A}}
            }
            {
                \vlhtr{\pi'}{\seq{\Gamma}{A}}
            }
        }
        \text{ or }
        \vlderivation
        {
            \vlin{\rr {\boxrule[\pbang{}]}}{}
            {
                \seq{\Gamma}{\just{\pbang{t}}{\just{t}{A}}}
            }
            {
                \vlhtr{\pi'}{\seq{\Gamma}{\just{t}{A}}}
            }
        }
    $$
    then we set $\tilde{\pi} = \pi'$ noting the conditions are met.

    If the preceding rule of $\pi$ is $\rr {\boxrule[\pdot{}{}]}$:
    $$
        \vlderivation
        {
            \vliin{\rr {\boxrule[\pdot{}{}]}}{}
            {
                \seq{\Gamma}{\just{\pdot{s}{t}}{B}}
            }
            {
                \vlhtr{\pi_1}{\seq{\Gamma}{\just{s}{(A \IMP B)}}}
            }
            {
                \vlhtr{\pi_2}{\seq{\Gamma}{\just{t}{A}}}
            }
        }
        $$
    by the inductive hypothesis, there exists proofs $\tilde{\pi_1}$ of $\seq{\Gamma}{A \IMP B}$ and $\tilde{\pi_2}$ of $\seq{\Gamma}{A}$ with $\rank{\seq{\Gamma}{\ann{A \IMP B}}}{\tilde{\pi_1}} < \rank{\seq{\Gamma}{\ann{\just{s}{(A \IMP B)}}}}{\pi_1}$, $\rank{\seq{\Gamma}{\ann{A}}}{\tilde{\pi_2}} < \rank{\seq{\Gamma}{\ann{\just{t}{A}}}}{\pi_2}$, $\cutrank{\tilde{\pi_1}} \leq \max{\cutrank{\pi_1}}{\rank{\seq{\Gamma}{\ann{\just{s}{(A \IMP B)}}}}{\pi_1} - 1}$ and
    $\cutrank{\tilde{\pi_2}} \leq \max{\cutrank{\pi_1}}{\rank{\seq{\Gamma}{\ann{\just{t}{A}}}}{\pi_2} - 1}$.

    We construct the proof $\pi'$ from the following
    $$
        \vlderivation
        {
            \vliq{\cont}{}
            {
                \seq{\Gamma}{B}
            }
            {
                \vliin{\cut}{}
                {
                    \seq{\Gamma, \Gamma}{B}
                }
                {
                    \vlhtr{\tilde{\pi_1}}{\seq{\Gamma}{A \IMP B}}
                }
                {
                    \vliin{\cut}{}
                    {
                        \seq{\Gamma, A \IMP B}{B}
                    }
                    {
                        \vlhtr{\tilde{\pi_2}}{\seq{\Gamma}{A}}
                    }
                    {
                        \vliin{\lr \IMP}{}
                        {
                            \seq{A, A \IMP B}{B}
                        }
                        {
                            \vlin{\gid}{}
                            {
                                \seq{A, B}{B}
                            }
                            {
                                \vlhy{}
                            }
                        }
                        {
                            \vlin{\gid}{}
                            {
                                \seq{A}{A}
                            }
                            {
                                \vlhy{}
                            }
                        }
                    }
                }
            }
        }
    $$
    noting that $\rank{\seq{\Gamma}{\ann{B}}}{\pi'} = \deg{B} < \deg{\just{\pdot{s}{t}}{B}} \leq \rank{\seq{\Gamma}{\ann{\just{\pdot{s}{t}}{B}}}}{\pi}$ (by \cref{lem:rankfacts}) and the cuts introduced to the system will be of most $\rank{\seq{\Gamma}{\ann{A \IMP B}}}{\tilde{\pi_1}}, \rank{\seq{\Gamma}{\ann{A}}}{\tilde{\pi_2}} < \rank{\seq{\Gamma}{\ann{\just{\pdot{s}{t}}{B}}}}{\pi}$.

    If the preceding rule of $\pi$ is $\rr{\boxrule[\pl{}]}$:
    $$
        \vlderivation
        {
            \vlin{\rr{\boxrule[\pl{}]}}{}
            {
                \seq{\Gamma}{\just{\pl{t}}{(A \IMP B)}}
            }
            {
                \vlhtr{\pi'}{\seq{\Gamma, \just{x}{A}}{\just{t}{B}}}
            }
        }
    $$
    by the inductive hypothesis there exists a proof $\tilde{\pi'}$ of $\seq{\Gamma, \just{x}{A}}{B}$ with $\rank{\seq{\Gamma, \just{x}{A}}{\ann{B}}}{\tilde{\pi'}} < \rank{\seq{\Gamma, \just{x}{A}}{\ann{\just{t}{B}}}}{\pi'}$ and $\cutrank{\tilde{\pi'}} \leq \max{\cutrank{\pi'}}{\rank{\seq{\Gamma, \just{x}{A}}{\ann{\just{t}{B}}}}{\pi'}-1}$.

    Construct $\pi'$ by
    $$
        \vlderivation
        {
            \vlin{\rr \IMP}{}
            {
                \seq{\Gamma}{A \IMP B}
            }
            {
                \vliin{\cut}{}
                {
                    \seq{\Gamma, A}{B}
                }
                {
                    \vlin{\rr {\boxrule[\ ]}}{}
                    {
                        \seq{A}{\just{x}{A}}
                    }
                    {
                        \vlin{\gid}{}
                        {
                            \seq{A}{A}
                        }
                        {
                            \vlhy{}
                        }
                    }
                }
                {
                    \vlhtr{\tilde{\pi'}}{\seq{\Gamma, \just{x}{A}}{B}}
                }
            }
        }
    $$
    noting that the only new cut introduced will be of $\rank{\seq{\Gamma, \just{x}{A}}{B}}{\tilde{\pi'}} < \rank{\seq{\Gamma}{\ann{\just{\pl{t}}{(A \IMP B)}}}}{\pi}$.

    If the preceding rule of $\pi$ is $\rr{\boxrule[\pael{}]}$:
    $$
        \vlderivation
        {
            \vlin{\rr{\boxrule[\pael{}]}}{}
            {
                \seq{\Gamma}{\just{\pael{t}}{A}}
            }
            {
                \vlhtr{\pi'}{\seq{\Gamma}{\just{t}{(A \AND B)}}}
            }
        }
    $$
    then by the inductive hypothesis there exists a proof $\tilde{\pi'}$ of $\seq{\Gamma}{A \AND B}$ similarly with the necessary conditions.

    Construct $\pi'$ from
    $$
        \vlderivation
        {
            \vliin{\cut}{}
            {
                \seq{\Gamma}{A}
            }
            {
                \vlhtr{\tilde{\pi'}}{\seq{\Gamma}{A \AND B}}
            }
            {
                \vlin{\lr \AND}{}
                {
                    \seq{A \AND B}{A}
                }
                {
                    \vlin{\gid}{}
                    {
                        \seq{A, B}{A}
                    }
                    {
                        \vlhy{}
                    }
                }
            }
        }
    $$
    which will satisfy the necessary conditions.

    The case for $\rr{\boxrule[\paer{}]}$ is similar.

    If the preceding rule of $\pi$ is $\rr{\boxrule[\pai{}{}]}$:
    $$
        \vlderivation
        {
            \vliin{\rr{\boxrule[\pai{}{}]}}{}
            {
                \seq{\Gamma}{\just{\pai{s}{t}}{(A \AND B)}}
            }
            {
                \vlhtr{\pi_1}{\seq{\Gamma}{\just{s}{A}}}
            }
            {
                \vlhtr{\pi_2}{\seq{\Gamma}{\just{t}{B}}}
            }
        }
    $$
    then by the inductive hypotheses, there exista proofs $\tilde{\pi_1}$ of $\seq{\Gamma}{A}$ and $\tilde{\pi_2}$ of $\seq{\Gamma}{B}$ with the necessary conditions.

    Construct $\pi'$ from
    $$
        \vlderivation
        {
            \vliin{\rr \AND}{}
            {
                \seq{\Gamma}{A \AND B}
            }
            {
                \vlhtr{\tilde{\pi_1}}{\seq{\Gamma}{A}}
            }
            {
                \vlhtr{\tilde{\pi_2}}{\seq{\Gamma}{B}}
            }
        }
    $$
    which will satisfy the necessary conditions.
\end{proof}

\end{document}